\documentclass[letterpaper,twocolumn,10pt]{article}
\usepackage{usenix}
\usepackage{booktabs}   

\usepackage{tikz}
\usepackage{amsmath}
\usepackage{amssymb}
\usepackage{amsthm}
\usepackage{multirow}
\usepackage{caption}
\usepackage[ruled,vlined]{algorithm2e}
\usepackage{academicons}  
\usepackage{filecontents}
\usepackage{adjustbox}
\usepackage[available]{usenixbadges}
\definecolor{blue}{named}{black}
\definecolor{purple}{named}{black}

\begin{document}

\date{}
\title{\Large \bf Attesting Model Lineage by Consisted Knowledge Evolution \\ with Fine-Tuning Trajectory}
\author{
{\rm ~~~Zhuoyi Shang$^{1,2,3,\star}$, Jiasen Li$^{1,2,3,\star}$, Pengzhen Chen$^{1,2,3,\star}$,    }\\
{\rm Yanwei Liu$^{ 1,3,\dagger}$, Xiaoyan Gu$^{1,2,3,\dagger}$, Weiping Wang$^{ 1 }$ }\\
$^{1}$Institute of Information Engineering, Chinese Academy of Sciences, Beijing, China\\
$^{2}$ School of Cyber Security, University of Chinese Academy of Sciences, Beijing, China \\
$^{3}$State Key Laboratory of Cyberspace Security Defense, Beijing, China\\
{\rm \{shangzhuoyi,~lijiasen,~chenpengzhen,~liuyanwei,~guxiaoyan,~wangweiping\}@iie.ac.cn}
} 
\maketitle
\begin{abstract}
The fine-tuning technique in deep learning gives rise to an emerging lineage relationship among models. This lineage provides a promising perspective for addressing security concerns such as unauthorized model redistribution and false claim of model provenance, which are particularly pressing in \textcolor{blue}{open-weight model} libraries where robust lineage verification mechanisms are often lacking. Existing approaches to model lineage detection primarily rely on static architectural similarities, which are insufficient to capture the dynamic evolution of knowledge that underlies true lineage relationships. Drawing inspiration from the genetic mechanism of human evolution, we tackle the problem of model lineage attestation by verifying the joint trajectory of knowledge evolution and parameter modification. To this end, we propose a novel model lineage attestation framework. In our framework, model editing is first leveraged to quantify parameter-level changes introduced by fine-tuning. Subsequently, we introduce a novel knowledge vectorization mechanism that refines the evolved knowledge within the edited models into compact representations by the assistance of probe samples. The probing strategies are adapted to different types of model families. These embeddings serve as the foundation for verifying the arithmetic consistency of knowledge relationships across models, thereby enabling robust attestation of model lineage. Extensive experimental evaluations demonstrate the effectiveness and resilience of our approach in a variety of adversarial scenarios in the real world. Our method consistently achieves reliable lineage verification across a broad spectrum of model types, including classifiers, diffusion models, and large language models.
\end{abstract}

{
  \renewcommand{\thefootnote}{\fnsymbol{footnote}}
  \setcounter{footnote}{0}

  \footnotetext{$\star$~These three authors contributed equally to this work.}
  \footnotetext{$\dagger$~Corresponding author.}
}

\section{Introduction}
The rapid advancement of deep learning (DL) has been largely driven by a paradigm shift from traditional \emph{training-from-scratch} approaches to the \emph{pre-training–fine-tuning} paradigm. This shift enables efficient domain adaptation, particularly for large-scale models, by allowing pre-trained models to be minimally fine-tuned for alignment with specific tasks or human preferences. As a result, fine-tuning introduces a new form of dependency among models, known as \emph{model lineage}~\cite{Lineage}, where multiple rounds of successive fine-tuning create inter-generational relationships between models.
\begin{figure}
	\centering
    	\includegraphics[width=.90\linewidth]{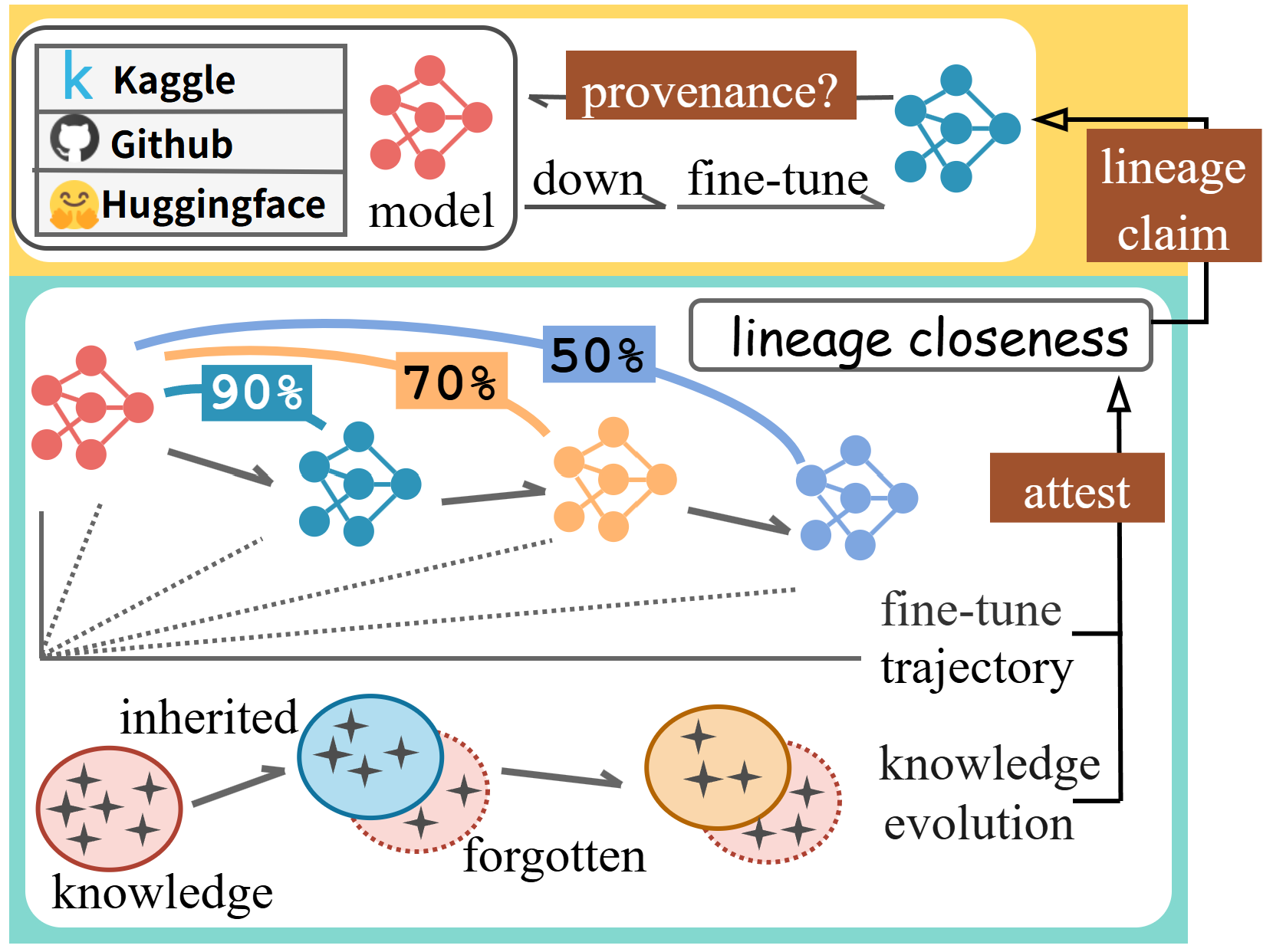}
	\caption{
		The knowledge evolution mechanism in model fine-tuning. The consistency between the knowledge of parent model and the total knowledge inherited and discarded during the evolution process is used to attest the lineage relationship.}
	\label{intro1}
\end{figure}

However, alongside the benefits of fine-tuning, new security challenges have emerged: (i) unauthorized redistribution via derivative models~\cite{FT-shield}, (ii) adversarial fine-tuning for harmful or unethical purposes~\cite{adversarialfinetuning}, and (iii) false claim of model provenance~\cite{fasleclaim}, especially in commercial settings. The rise of open-weight model repositories (e.g., GitHub and Hugging Face\cite{transformers}) ~\cite{modelzoo,modelzoo-2,hugginggpt,hugging-reuse}, has further amplified these risks, as current platforms often lack effective lineage attestation mechanisms. This creates opportunities for both inadvertent misuse and deliberate obfuscation of true model ancestry.

Regarding these security concerns, it is critically important to develop robust forensic techniques for model lineage attestation. Such tools are essential for intellectual property (IP) protection~\cite{SSLGuard,IPGuard,KRM,LLMCG}, model authentication~\cite{authentication}, and defect tracing~\cite{orignllms}. They can provide conclusive evidence of inheritance between the base and derived models, serving as a foundation for legal and ethical accountability. For example, if a model owner suspects that a publicly released model is derived from his proprietary model, he may first conduct behavioral comparisons through black-box access~\cite{baihe}. \textit{Upon gathering sufficient suspicion, the owner could request white-box access through legal channels to extract unique model-specific identifiers and assert ownership.} However, existing authentication methods struggle with model lineage verification due to the disruptive nature of fine-tuning on both parameters and behavioral traits.

Furthermore, recent advancements in model lineage detection~\cite{Lineage,NP} have yet to yield unambiguous, verifiable evidence of lineage. These approaches remain vulnerable to dynamic adversarial spoofing, as they primarily rely on static 
architectural similarities (feature-similarities or parameter similarities) between models. 
Such structural representations fail to capture the deep, dynamic transformations in a model’s internal knowledge representation during fine-tuning. As a result, these approaches often yield suboptimal attribution performance. Indeed, fine-tuning is not merely a process of parameter adjustment—it reflects a continuous evolution of model knowledge, guided by the semantics of the fine-tuning dataset. \textit{This highlights two key challenges for lineage attestation: (1) how to effectively characterize the intrinsic variations in model knowledge throughout the fine-tuning process, and (2) how to leverage the dynamic process of knowledge transformation for reliable and interpretable lineage attribution.}

To address the above two challenges, we propose to attest model lineage by exploring the knowledge evolution in a fine-tuning process. Our approach is inspired by biological evolution, which serves as a powerful metaphor for model lineage. In human genetics, genes—carriers of hereditary information—are inherited from parents, yet evolution entails a dynamic process where some genetic traits are preserved across generations while others are gradually forgotten. As illustrated in Fig.~\ref{intro1}, in the context of model evolution, knowledge functions like genes: the fine-tuned model inherits knowledge from its parent, while the fine-tuning process inevitably alters or forgets certain portions of the original knowledge. 
Therefore, lineage similarity can be measured by examining the consistency between the knowledge of parent model and the totality of inherited and discarded knowledge during the evolution process.

A model’s knowledge is structurally encoded in its parameters, enabling us to capture the knowledge variations in fine-tuning process at the parameter level. Fortunately, model editing techniques~\cite{taskari} allow the parameters of a fine-tuned (child) model to be decoupled into those of its parent model and a task vector derived by fine-tuning. The additional task vector can be instantiated as a neural network, referred to an evolution model, delineating knowledge variations caused by fine-tuning. 
\textcolor{blue}{
Then we can encode the intermediate or final responses of both the evolution model and the original/fine-tuned models into knowledge vectors using probing datasets.  The probing strategies are tailored to different model structures, including decision-boundary probing for classification models, distributional sampling for diffusion models, and  response-based probing for LLMs. }
Our intuition is that a model’s behavior on specific inputs reflects its internal knowledge, enabling an arithmetic way to measure the consistency of knowledge evolution with fine-tuning trajectory. 

In summary, we first  explicitly parameterize the knowledge modifications introduced during fine-tuning. Building on this, we propose an effective knowledge vectorization mechanism  to embed the edited knowledge into a shared latent space, that offers a principled foundation for model lineage attestation, where the consistency between the fine-tuning trajectory and the knowledge evolution path forms the necessary condition for validating model lineage relationships.

Our main contributions are summarized as follows:
\begin{itemize}
     \item Inspired by the genetic mechanism of human evolution, we uncover for the first time the knowledge evolution mechanism in the model fine-tuning process, grounding the effective model lineage verification in internal representational evidence.
    \item We propose a novel model lineage attestation methodology by exploring consisted knowledge evolution with fine-tuning trajectory.
    To verify the knowledge consistency, we devise a novel vectorization mechanism to embed the models' knowledge in a shared latent space, embracing a variety of model types. 
    Our methodology can capture the intrinsic genealogical links among models, outperforming those with only comparing superficial resemblance among models.  
    
	
	
	\item We perform comprehensive experimental evaluations for two prominent adversarial scenarios. Our method demonstrates robust lineage verification across a wide spectrum of model types, including classification models, diffusion models, and \textcolor{blue}{ diverse large language models, }
    showing strong resilience against real-world adversarial manipulations.
\end{itemize}

\section{Related Work}

\subsection{Neural Network IP Protection}
To detect unauthorized copies of protected models, several techniques are commonly employed for DL model IP protection. One kind of methods is model watermarking\cite{watermark2,watermarking,watermark3,Turning-watermark}, which involves embedding a detectable identifier into the neural network through processes such as fine-tuning or retraining.  They either embed secret information directly into the model weights or introduce backdoor attacks to manipulate training samples to deliberately alter the decision boundaries of the neural network to protect its intellectual property. However, both methods require modifying the original model's parameters. The former is vulnerable to changes caused by attacks such as pruning and fine-tuning, while the latter lacks robustness when facing downstream tasks that require altering decision boundaries, such as fine-tuning.
Another category of technique is model fingerprinting~\cite{IPGuard,ijcai-figure,figureprin-cvpr,plugm,llmmap}, which aims to identify a set of specific samples that trigger similar responses in both the original and pirated models. The classification boundary of a DNN classifier serves as its unique representation. Leveraging this insight, IPGuard~\cite{IPGuard} identifies data points close to the model owner's classification boundary and uses them to create a distinctive fingerprint for the classifier. If another DNN classifier produces matching predictions for the majority of these fingerprint data points, it is deemed to be an unauthorized copy of the original classifier. From the perspective of knowledge transfer in DL models, Tian et al.~\cite{KRM} observed the knowledge transferred from training data to a DL model can be uniquely represented by the model's decision boundary, and they further proposed that the distance from the central training samples to the boundary decision samples can serve as the model's fingerprint for its IP audit.  Although their scheme shows resistance to some fine-tuning attacks and is more robust than IPGuard\cite{IPGuard}, it remains vulnerable when fine-tuning fails to restore the model's decision boundary.

\subsection{Fine-tuning Relationship Detection}
At present, there is relatively limited research on detecting fine-tuning relationship among models. Typically, Neural Lineage~\cite{Lineage} and Neural Phylogeny~\cite{NP} are used to delineate the fine-tuning relationships. As a seminal work, Yu et al. \cite{Lineage} have defined the following branches of study on neural lineage:
neural lineage refers to the relationship between a parent model and its child model, specifically identifying which parent model a given child model has been fine-tuned from.  
In neural lineage, fine-tuning is defined as a standard transfer learning practice: adapting a model from the original dataset to a new task dataset, with the requirement that the fine-tuned model performs well on the new task. The neural lineage detection in \cite{Lineage} relies on access to fine-tuning datasets for validation. This requirement is impractical in real-world model protection scenarios, as malicious publishers are unlikely to disclose their datasets used for fine-tuning. To further identify the direction of  fine-tuning relationship, the concept of neural phylogeny detection is then proposed in the succeeding work of Yu et al. \cite{NP}. This method relies on static architectural similarities among fine-tuned models, which struggles to grasp the dynamic knowledge associations among models that truly characterizes the lineage relationships. In this work, for clarity in our discussion, we use the notation of "model lineage" to uniformly describe fine-tuning relationships between models, encompassing directed fine-tuning dependencies.

\subsection{Task Arithmetic for Model Editing}
Task Arithmetic~\cite{taskari,tangentspace,2024-model} has been explored as a powerful tool for analyzing and leveraging task relationships in multi-task learning.  
Its core idea is to treat tasks as vectors and explore how their combinations influence the behavior of the model. 
  One key advantage of task arithmetic lies in its ability to decouple knowledge components about the multi-task models. 
The performances of pre-trained models on different datasets can be independently modulated by adding or removing task vectors \cite{taskari}, where "addition" improves task-specific performance and "removal" may lead to forgetting. This model editing property provides an arithmetic approach to analyze fine-tuning dependencies between models, \textcolor{blue}{directly forming} technical basis for model lineage verification in the paper.

\section{\textcolor{blue}{Problem Statement and Threat Model}}
\label{threatmodel}

\subsection{Problem Statement}
\textcolor{blue}{
The purpose of Model Lineage Attestation (MLA) for DL models is to determine whether a given suspect model $f_C$ is a fine-tuned version of a parent  model $f_P$.  When the owner $P$ of the parent model $f_P$ discovers a suspected model $f_C$, he can ask a trusted authentication authority to verify the lineage relationship. The  attestation task is therefore to provide a reliable criterion that distinguishes models belonging to the same lineage from those that do not.}

\textcolor{blue}{
Specifically, the goal of MLA is to construct an appropriate \emph{lineage-aware metric space} in which models from the same family exhibit consistently high lineage similarity, while models from different families exhibit low lineage similarity. 
Under a decision threshold $T$, the attestation framework aims to induce a similarity metric $\mathcal{S}(\cdot,\cdot)$ space such that $f_C$ and $f_P$ belong to the same family if $\mathcal{S}(f_C,f_P) \;\ge\; T$.}

\subsection{Threat Model}

\textcolor{blue}{
We consider model lineage attestation under two key threat scenarios that focus on adversarial attempts to manipulate model lineage. They are likely often encountered in open-weight model hubs and we illustrate them in Fig.~\ref{fig:example}. To detail which assets each party 
may access or manipulate, a fine-grained summary of the capabilities of each role under these threat scenarios is offered in Table~\ref{tab:threat-scenarios}.}
\begin{figure}[htbp]
	\centering
	\includegraphics[width=.85\linewidth]{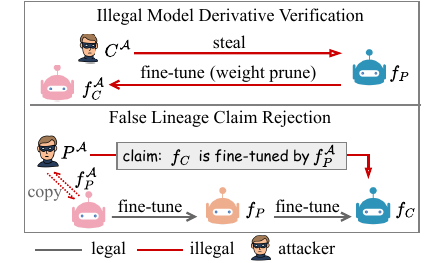}
	\caption{Two key threat scenarios for model lineage attestation. In the illegal derivation scenario, the adversary  steals the victim model $f_P$ and derives a new model $f_C^\mathcal{A}$  with direct fine-tuning or pruning followed by fine-tuning. In the false lineage claim scenario, the attacker attempts to present a closely related model $f_P^\mathcal{A}$ to falsely declare that a victim model $f_C$ was derived from it.
	}
	\label{fig:example}
\end{figure}

\begin{table*}[t]
\caption{\textcolor{blue}{Roles and accessible assets under different threat scenarios. In the Illegal Model Derivative setting, the adversary derives and modifies $f_C$ 
from a stolen parent model without access to the key evidence. 
In the False Lineage Claim setting, the adversary attempts to fabricate a lineage relationship by supplying forged evidence.}}
\centering
\renewcommand{\arraystretch}{1.2}
\setlength{\tabcolsep}{8pt}
\begin{adjustbox}{max width =0.9\textwidth}
\begin{tabular}{lcccc}
\toprule
 & $f_P$ & $f_C$ &Attest evidence \\
\midrule
Defender
& Submitted 
& Submitted 
& Trusted \& Held  \\
\midrule
Adversary (Illegal Model Derivative)
& \textcolor{red}{\Large \textbf{$\times$}} 
& Derived \& Modified 
& \textcolor{red}{\Large \textbf{$\times$}} \\
\midrule
Adversary (False Lineage Claim)
& Forged & \textcolor{red}{\Large \textbf{$\checkmark$}}
& Forged \\
\bottomrule
\end{tabular}
\label{tab:threat-scenarios}
\end{adjustbox}
\end{table*}

\subsubsection{Illegal Model Derivative Verification}
\label{sec:threat1}
\textcolor{blue}{
A model owner trains and deploys a proprietary model $f_P$, while an adversary $C_\mathcal{A}$ may steal it and release a tampered or fine-tuned variant $f_C^\mathcal{A}$. The model owner therefore requires a reliable attestation mechanism to verify lineage relationships.}

\textcolor{blue}{
\textit{Adversary’s Assumptions.}  
The adversary is assumed to have \textit{white-box access} to the victim model $f_P$, including its architecture and parameters, with sufficient resources to perform fine-tuning. For complex models such as large language models or diffusion models, they may apply advanced fine-tuning techniques such as DreamBooth~\cite{dreambooth}, LoRA~\cite{lora}, or parameter merging to modify the stolen model.  Moreover, the adversary may be aware of the lineage attestation technique  and further adopt evasion strategies like weight pruning or fine-tuning across multiple generations to obscure lineage. However, $C_\mathcal{A}$ does \textit{not} have access to the original training data or configurations of $f_P$, which constrains their ability to perfectly reproduce the model’s provenance.}

\textcolor{blue}{
\textit{Defender’s Assumptions.}  
We assume the defender is a trusted lineage attestation authority whose goal is to verify whether a suspicious model $f_C^\mathcal{A}$ is derived from a proprietary model $f_P$. The authority may receive $f_P$ and any auxiliary evidence(probe dataset and the initialize parameters) provided by the model owner. 
The defender can obtain the parameters of $f_C^\mathcal{A}$ legally  by filing a complaint to a judicial body. 
According to the principle ``he who asserts must prove,'' it is reasonable for the claimant to provide auxiliary evidence to substantiate ownership. 
 The authority has no knowledge of the adversary’s fine-tuning strategies or modifications. Under these assumptions, the attestation mechanism can determine the lineage relationship between $f_C^\mathcal{A}$ and $f_P$.}

\subsubsection{\textcolor{blue}{False Lineage Claim Rejection}}
\label{sec:threat2}
\textcolor{blue}{An adversary $P_\mathcal{A}$ may attempt to falsely claim that an existing victim model $f_C$ was derived from their model $f_P^\mathcal{A}$, aiming to gain undeserved credit. Consequently, a fair lineage attestation authority would reject any unfounded claims.}

\textcolor{blue}{
\textit{Adversary Assumptions.}
The adversary may attempt to gain undeserved credit by presenting a model $f_P^\mathcal{A}$ as the parent of a victim model $f_C$. To do so, they may steal or copy a close-lineage model of $f_C$ within the same family, such as a grandparent model, along with its corresponding auxiliary evidence, such as a probe dataset or initialization parameters.  The adversary has white-box access to $f_C$, but lacks knowledge of its training configuration. They cannot access the true parent model $f_P$  of $f_C$ and are unaware of the lineage closeness between $f_P^\mathcal{A}$ and $f_C$.  Parent-side artifacts are 
maintained by the model platforms, which prevents the
adversary from altering or fabricating these artifacts, including the initialization $\theta_0$.
  }

\textcolor{blue}{
\textit{Defender’s Assumptions.} 
The defender (lineage attestation authority) is aware of $f_C$'s parameters and also its training configuration, and he has also white-box access to $f_P^\mathcal{A}$ legally.
By comparing the lineage similarity between $f_P^\mathcal{A}$ and $f_C$, the lineage verification can identify assertions where $f_C$ and $f_P^\mathcal{A}$ do not share a direct parent-child relationship.}

\section{Knowledge Evolution in Model Fine-tuning}
Fine-tuning a neural network encapsulates knowledge into parameters through a process of selection and consolidation
 in the training, where the child model selectively inherits knowledge from the parent model, guided by the objective of the fine-tuning task. Critically, model editing theory reveals that fine-tuned parameters
remain task-disentangled—enabling us to systematically analyze the knowledge evolution process through parameter dynamics.

\textbf{Characterizing Knowledge Evolution through Task Editing.} \textcolor{blue}{
Since neural network parameters can be decomposed into distinct components corresponding to different training tasks~\cite{taskari,ICLRnonlinear}, the fine-tuning trajectory can be similarly characterized by changes in these parameters, even under nonlinear optimization dynamics~\cite{ICLRnonlinear}.}
Given a submodel \( f_C(\cdot;\theta_C) \) with parameters \(\theta_C\), and its parent model \( f_P(\cdot;\theta_P) \) with parameters \(\theta_P\), the parameter difference from $\theta_C$ to $\theta_P$ during fine-tuning can be denoted as the task vector \(\Delta = \theta_P - \theta_C\). \(\Delta\) effectively encodes the model evolution process introduced by retraining, which provides a solution to quantify the unknown knowledge evolution during fine-tuning process.  
Therefore, we can find a surrogate model $f_\Delta(\cdot;\theta_\Delta)$ (called evolution model) with parameters $\theta_\Delta$ 
that encodes the evolutionary knowledge offset during the fine-tuning process. This model is constructed within the same search space spanned by the initial parameters $\theta_0$ of $f_P$ as:
\begin{equation}  
	\label{theta_delta}
         \vspace{-0.1cm}
	\theta_\Delta = \theta_0 + \theta_P -\theta_C.
     \vspace{-0.1cm}
\end{equation} 

Therefore, for a submodel \(f_C \) that is fine-tuned from $f_P$, \(f_C \)'s behavior explained from the perspective of the parent model's knowledge is influenced by two parts: one is the direct knowledge transfer from the parent model, the other is the unknown knowledge evolution during the fine-tuning process. Thus, from the perspective of knowledge inheritance, knowledge consistency among these models can be expressed as:
 \begin{equation}
 	\label{kncons}
 	\mathcal{K}_P(f_P) = \mathcal{K}_P(f_C) + \mathcal{K}_P(f_\Delta), 
 \end{equation}
where \(\mathcal{K}_P\) denotes the specific knowledge encoded in \( f_P \), which is tied to its parameters. 
Specifically, \(\mathcal{K}_P(f_P)\) refers to $f_P$'s understanding of the learned knowledge from $f_P$'s training dataset $D_P$, and 
\(\mathcal{K}_P(f_C)\) indicates the child model's observation of the parent model's knowledge.  Due to the impact of fine-tuning, the child model inevitably forgets or alters the knowledge of the parent model,
which is denoted as \(\mathcal{K}_P(f_\Delta)\).

\textbf{Consistent Knowledge Evolution with Fine-tuning Trajectory.} 
No matter how descendants evolve, a family lineage always traces back to the same ancestral origin, as encoded in their shared genes. Similarly, if a child model $f_C$ is fine-tuned from a parent model $f_P$, no matter what downstream task is used for fine-tuning, its initialization parameters are necessarily determined by the parent model, even minor changes in the initialization parameters can lead to divergent convergence paths ~\cite{init}. Therefore, the training trajectory of a model family is unique and difficult to replicate.

\textcolor{blue}{
 Although models from distinct lineages may exhibit similar knowledge due to potential overlaps in their training datasets, Theorem~\ref{thealign} claims that the \(\mathcal{K}_P\) knowledge shift that consists with the fine-tuning trajectory described by model parameters modification is an inherent property of model fine-tuning, and it accordingly provides a intrinsic signature for the model lineage attestation.}

\newtheorem{theorem}{Theorem}   
\begin{theorem}  
\label{thealign}  
For a pre-trained model \(f_P\), a fine-tuned model \(f_C\), and the generated model \(f_{\Delta}\) calculated through task arithmetic by Eq.~(\ref{theta_delta}), that the \(\mathcal{K}_P\) knowledge shift from \(f_P\) to \(f_C\) is consistent with the evolved \(\mathcal{K}_P\) knowledge in \(f_{\Delta}\) is a necessary condition for confirming that \(f_C\) is fine-tuned from \(f_P\).  
\end{theorem}  
\begin{proof}  
\textcolor{blue}{
For any true descendant, the parameter update $\Delta=\theta_P-\theta_C$ lies on the task-optimization trajectory and thus satisfies (or approximately satisfies) the task-aligned weight-decoupling structure required for task arithmetic \cite{taskari,ICLRnonlinear,2024-model}.
Assume, for contradiction, that a non-descendant model $f_C'$ with weight $\theta'_{C}$ also satisfies Eq.~(\ref{kncons}). Since the parameter difference $\Delta'=\theta_P-\theta'_{C}$ is not generated by fine-tuning the parent on any task, it is not on the task optimization path. Thus, it cannot satisfy the weight-decoupling property, as it violates the fundamental assumption of task arithmetic. Accordingly, $f'_{\Delta}$ fails to act as a valid evolution model explaining the parent knowledge discrepancy among $f_C'$ and $f_P$, making Eq.~(\ref{kncons}) impossible. Thus, Eq.~(\ref{kncons}) holds only if $f_C$ is a fine-tuned descendant of $f_P$.}
\end{proof}



\textcolor{blue}{Importantly, this argument extends beyond the idealized linear fine-tuning case.  Realistic fine-tuning may update parameters non-linearly, but it preserves the alignment of $\Delta$ with the parent-task knowledge direction, which is the structural property required by Eq.~(\ref{kncons}). In contrast, $\Delta'$ from a non-descendant achieves no such task-induced alignment and cannot satisfy even approximate weight-decoupling. Because the directional alignment is maintained only for true descendants and absent for non-descendants, Eq.~(\ref{kncons}) continues to hold under realistic fine-tuning dynamics.}



\section{Model Lineage Attestation}

\begin{figure*}
\centering
	\begin{minipage}[ht]{.99\linewidth}
		\centering
        \includegraphics[width=.95\linewidth]{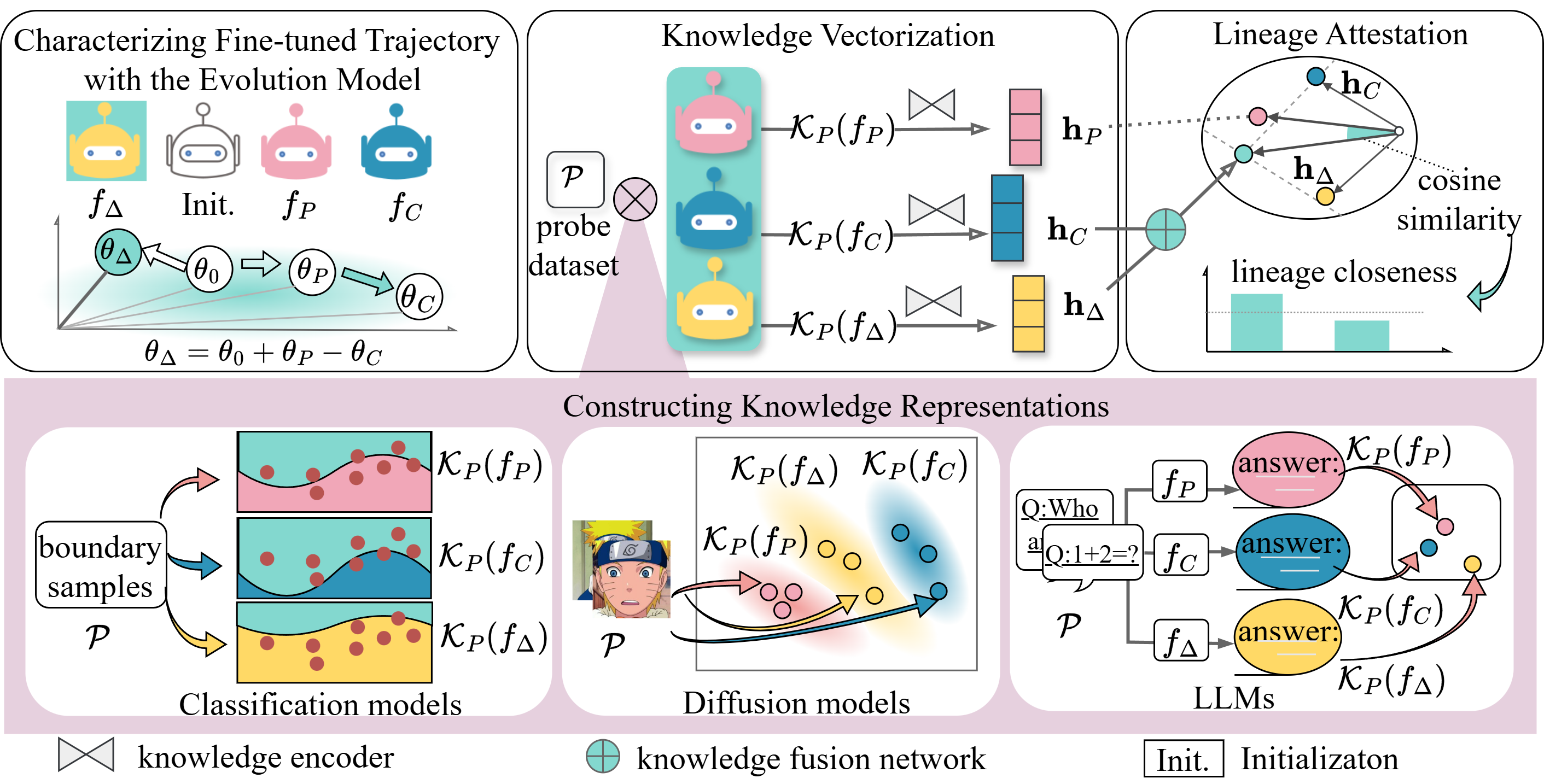}
	\end{minipage} 
	\caption{The proposed model lineage attestation (MLA) framework consists of three key components. 
        First, given a parent model $f_P$ and its fine-tuned model $f_C$, MLA constructs an evolution model $f_{\Delta}$ to characterize knowledge evolution during fine-tuning. 
        Then, all three models are encoded into vector representations ($\mathbf{h}_P$, $\mathbf{h}_{\Delta}$, $\mathbf{h}_C$) through a shared knowledge encoder.
        Finally, the lineage attestation is performed by evaluating the knowledge consistency condition in Equation (\ref{kncons}).}
	\label{method}
\end{figure*}

To obtain a more accurate lineage similarity $\mathcal{S}(f_C,f_P)$, we devise a model lineage attestation framework by examining the consistency of knowledge evolution along the fine-tuning trajectory, as shown in Fig.~\ref{method}. We first parameterize the fine-tuning trajectory as a virtual neural network inspired by task arithmetic theory. Then, we develop an effective knowledge vectorization method, embedding neural network knowledge into a shared latent space for exploring knowledge consistency. The resulting vector space enables us to reliably verify the lineage relationship \textcolor{blue}{between the original and fine-tuned models} by comparing the resulting knowledge consistency score with a specific threshold~$T$.

\subsection{Constructing Evolution Model $f_\Delta$}  
To depict the process of knowledge evolution for model lineage attestation, the prerequisite is to construct an evolution model $f_\Delta$ with model editing theory. To this end, we reinterpret the parameter transformation between two models as a neural network $f_\Delta(\cdot;\theta_\Delta)$ in terms of Eq.~(\ref{theta_delta}). Thus, $f_\Delta$ effectively characterizes the fine-tuned trajectory, and the forward inference response on \(D_p\) can be leveraged to probe the changes in its knowledge about \(D_p\) induced by fine-tuning.

\subsection{Model Knowledge Vectorization}
To refine the parametric knowledge in the model, we propose a novel probing-based knowledge vectorization methodology that captures the model's decision or functional behaviors associated with the probing dataset. This methodology offers an effective mechanism for representing the knowledge in model evolutions that strongly associates with parent model's training dataset. Given the variations across different model types, we devise specialized knowledge vectorization approaches tailored to each type of model. \textcolor{blue}{
 For classification models, knowledge resides in structured decision boundaries learned from the training dataset. For diffusion models, knowledge is the learned data distribution. We probe by sampling from this distribution to test generative fidelity. Due to tokenizer-induced misalignment, we probe LLM knowledge via behavior-level responses on multi-domain benchmarks.
}
Despite differences in model types, the key steps in vectorizing models remain the same: (1) creating a probing dataset derived from the parent model’s training data, and (2) encoding the model’s decision/response patterns from this dataset into knowledge vectors.

\subsubsection{Knowledge Vectorization of Classification Models}
The knowledge learned by a classification model during training is implicitly captured by its decision boundaries, which delineate class-specific regions in the input space. 

\textbf{Constructing Representative Probe Samples.}
The decision boundary of a neural network encodes its unique learned knowledge, such that samples lying close to this boundary offer a concise representation of the knowledge it has acquired~\cite{KRM}.
Since fine-tuning typically alters the parent model’s decision boundary, the child model’s responses to these boundary samples naturally reflect the changes induced by fine-tuning. In particular, by examining the embeddings or predictions on the same probe dataset, we can infer how the original knowledge has been preserved, modified, or lost.  On the other hand, centroid samples capture the core semantics of each class. Accordingly, we construct a relatively small probe dataset \(\mathcal{P}\) consisting of boundary samples and centroid samples derived from \(f_P\).

Assume the parent model \(f_P = g(\epsilon(\cdot))\) is a $k$-class DNN classifier, where \(\epsilon(\cdot)\) denotes the feature extractor and \(g(\cdot)\) denotes the classification head.  
Formally, let $ \{g_i\}$ denote the decision functions of the target classifier of class $i$, 
and a data point $x$ is on the target classifier’s classification boundary if at least two labels have the largest discrimination probability, i.e., $ g_j(x) = g_{j'}(x) \geq \max\limits_{i \neq j,j'} g_i(x)$, where $i,j,j'$ are the category indexes, and $g_i(x)$ is the probability that sample $x$ belongs to category $i$ ~\cite{IPGuard}. Given the  training dataset \( D_P \), we construct the probe sample set $\mathcal{P}$ as:
\begin{equation}
\vspace{-0.1cm}
	\mathcal{P} = \{x_0^0, x_0^1, x_0^2, \dots, x_i^i, x_i^j, \dots, x_k^k\},
\label{equation:input}
    \vspace{-0.1cm}
\end{equation}
where \( x_i^i \) is the centroid sample of class \( i \), and \( x_i^j \) (with \( i \ne j \)) denotes a boundary sample that transits from class \( i \) to \( j \).


%


%

\textbf{Extracting Feature-Encoded Knowledge via Decision Boundary Probes.}  To analyze how the parent model’s knowledge evolves through fine-tuning, we begin by examining the feature embeddings that the parent \(f_P\), the evolution model \(f_\Delta\) and the fine-tuned child \(f_C\) extract from the carefully chosen samples.
The feature embeddings produced for these samples reflect the model’s ability about the corresponding knowledge,
thereby providing an indirect yet effective means to analyze the specific knowledge relationship. 
Therefore, the parent model’s knowledge $\mathcal{K}_P(f_P)$ with respect to $\mathcal{P}$ is represented by:
\begin{equation}
\vspace{-0.1cm}
	\mathcal{K}_P(f_P) = \big\{\epsilon_P(x) \;\big|\; x \in \mathcal{P} \,\big\}.
    \vspace{-0.1cm}
\end{equation}

Fine-tuning typically alters a model’s decision boundary.  
We can thus capture the inherited knowledge of the child model by observing its behavior on the parent-derived probe set \(\mathcal{P}\).  
Specifically, the positions of the parent model’s boundary samples within the child model reflect how fine-tuning has modified the decision regions.   
Accordingly, the \(f_C\)’s inherited knowledge is represented by:
\begin{equation}
\vspace{-0.1cm}
\mathcal{K}_P(f_C) = \left\{\epsilon_C(x) \;\middle|\; x \in \mathcal{P} \right\},
\vspace{-0.1cm}
\end{equation}
where  \(\epsilon_C(\cdot)\) denote the child model’s feature extractor.  

Similarly, the evolution model \(f_\Delta\) characterizes the transformation from the parent model to the child model, and  
its behavior on the parent-derived probe set \(\mathcal{P}\) captures knowledge that has been lost or modified during fine-tuning.  
Formally, the knowledge representation $\mathcal{K}_P(f_\Delta)$ is defined as:
\begin{equation}
\vspace{-0.1cm}
\mathcal{K}_P(f_\Delta) = \big\{\epsilon_\Delta(x) \;\big|\; x \in \mathcal{P} \,\big\}.
\vspace{-0.1cm}
\end{equation}

\textbf{Encoding Feature Knowledge into Vectors.} 
To encode semantic dependencies across categories, we construct a knowledge encoder $\Psi$,  in which categories are modeled as sequences and each class occupies a unique position. 
Feature embeddings are first projected into a shared latent space via a multi-layer perceptron (MLP) block \( \mathcal{M} \). A Transformer encoder \( \mathcal{T} \) is then applied \( k \) times—once per class—taking the centroid and corresponding boundary samples as input to capture local inter-class relations. The last-token embedding from each sequence is used as the class-specific knowledge vector, and the final representation is obtained by averaging these vectors via a global pooling layer \( \mathcal{G} \).
This forms a hierarchical knowledge encoder denoted as $\Psi = \mathcal{G}(\mathcal{T}(\mathcal{M}(\cdot)))$.
Formally,  for model $f_F \in \{f_P, f_C, f_\Delta\}$, we then obtain the compact knowledge vectors by applying the hierarchical knowledge encoder $\Psi$:
\begin{equation}
	\label{eq8}
\mathbf{h}_F = \Psi(\mathcal{K}_P(f_F)), \quad F \in \{P, C, \Delta\}.
\end{equation}

\subsubsection{Knowledge Vectorization for  Diffusion Models and LLMs} 
During the forward inference process, the model encodes each input into a semantically rich embedding in its intermediate layers. For multi-domain models, such as diffusion models and large language models, which are trained on massive datasets, it is nearly impossible to access the full training data.  Following prior work~\cite{knowledgeconsis}, the model’s knowledge can be defined as the semantic concepts encoded in its feature representations of the probe inputs at a given layer, since these representations reflect what the model has learned to distinguish the query input. \textcolor{blue}{However, internal probes (logits or representations) become infeasible when child models alter tokenizers for LLMs, causing internal representation misalignment. We therefore evaluate the consistency of model knowledge by comparing their responses to the same set of probe inputs~\cite{LLM}, which serves as a robust proxy for internal knowledge.}

\textbf{Constructing Feature Knowledge Vector for Diffusion Models.} 
For diffusion models, we create a probing dataset, $\mathcal{P} = \{x_1, x_2, \dots, x_n\}$, which consists of $n$ randomly selected images from $D_P$. The generative capabilities of these models are fundamentally driven by the denoising backbone of the U-Net architecture~\cite{freeu}. Different diffusion models, however, often utilize U-Nets with significantly varied architectures or training objectives, with these differences being substantial across model families. These architectural and objective variations result in distinct intermediate feature representations when given the same input data. Therefore, we leverage the output features of the \texttt{UpBlock} layers within the U-Net as meaningful embeddings to capture this unique knowledge. Let $\epsilon(\cdot)$ denotes the feature extraction process from the final \texttt{UpBlock} layer of the models, we can instantiate the knowledge, resulting in $\mathcal{K}_P(f_\Delta)$, $\mathcal{K}_P(f_C)$, and $\mathcal{K}_P(f_P)$ as previously established.

To transform the extracted features into compact knowledge representations, we apply a knowledge encoder \( \Psi \) that aligns and aggregates feature vectors. Due to the spatial nature of extracted features, we first apply a convolutional layer \( \mathcal{C} \) to convert the 3D features into vectors, followed by global average pooling \( \mathcal{G} \), and an MLP block \( \mathcal{M} \) for alignment: $\Psi = \mathcal{M}(\mathcal{G}(\mathcal{C}(\cdot)))$.
Finally, the knowledge vectors $\mathbf{h}_P,\mathbf{h}_C, \mathbf{h}_\Delta$ for each model is obtained as Eq.~(\ref{eq8}).

\textbf{Acquiring Response-based Knowledge Vector for LLMs.} 
We estimate the model  knowledge by analyzing the responses of the parent model$f_P$, child model $f_C$, and evolution model $f_\Delta$ to the same inputs. 
The probe dataset is defined as \( \mathcal{P} = \{ x_1, x_2, \dots, x_n \} \), where each query \( x_i \) is used to elicit responses from the models \( f_P \), \( f_C \), and \( f_\Delta \).  To ensure diverse coverage  and comprehensively acquiring a model’s knowledge, $\mathcal{P}$ spans multiple domains—such as general Q\&A, mathematics, and reasoning. 
To account for the  randomness in LLM generation, 5 responses are collected for each question via independent inference.	
Therefore, 
for a model $f_F \in \{f_P, f_C, f_\Delta\}$, we define its knowledge representation over $\mathcal{P}$ as:
\begin{equation}
\vspace{-0.1cm}
	\begin{split}
		\mathcal{K}_P(f_F) &= \big\{f_P(x)^{(r)} \;\big|\; x \in \mathcal{P},\; r = 1,\dots,5 \big\}, \\
	\end{split}
    \vspace{-0.1cm}
\end{equation}
where $f_P(x)^{(r)}$ denotes the encoded embedding of the $r$-th response to query $x$ encoded by the parent model's encoder.

Since LLM progressively refines input representations through stacked layers, and the final hidden states produce a contextualized embedding that captures its understanding of the input prompt, which is subsequently decoded into a response. To obtain a concise representation of a model’s knowledge, we transform the generated responses to each input into continuous embeddings, which act as compact knowledge vectors. In detail, all responses are encoded using the parent model’s encoder $\epsilon_P$ to remove semantic differences across models and map them into a common embedding space, enabling us to better observe how knowledge from the parent model is preserved or transformed in the child and evolution models. For each model $f_F$,
we then aggregate the resulting embeddings using the knowledge encoder \(\Psi\) to obtain the final knowledge vector $\mathbf{h}_F$ of the model:
\[
\vspace{-0.1cm}
\mathbf{h}_F = \Psi(\epsilon_P(\mathcal{K}_P(f_F))),
\quad F \in \{P, C, \Delta\},
\vspace{-0.1cm}
\]
where $\Psi = \mathcal{G}(\mathcal{M}(\cdot))$ consisting of an MLP layer \( \mathcal{M} \) for refining and aligning the initial response-based embeddings and a global average pooling layer \( \mathcal{G} \) for aggregating them into the compact knowledge representation.


\subsection{Model Lineage Attestation}

In the vectorized metric space, the knowledge vectors \( \mathbf{h}_C \), \( \mathbf{h}_{\Delta} \), and \( \mathbf{h}_{P} \) for a fine-tuning process should align with Eq.~(\ref{kncons}), such that the fusion vector of \( \mathbf{h}_C \) and \( \mathbf{h}_\Delta \) is aligned with \( \mathbf{h}_P \).
However, practical limitations arise because of the imperfect knowledge extraction.
Firstly, from the perspective of task arithmetic, $f_\Delta$ is constructed along the fine-tuning direction guided by the fine-tuning dataset, and is therefore semantically misaligned with the probe dataset, since the latter is semantically aligned with the parent model.
Secondly, the evolution model $f_\Delta$  generated through model editing only approximates the model's optimization trajectory and does not perfectly represent the changed knowledge.
Consequently, the extracted feature embeddings may contain noise and fail to fully preserve the knowledge of the parent model.

\textbf{Constructing Knowledge Alignment Space.}
To mitigate these limitations and refine the approximation, we introduce a trainable knowledge fusion network $\Phi$, which is employed as a linear fully connected network, to further aggregate the knowledge in $\mathbf{h}_C$ and $\mathbf{h}_\Delta$ and enable knowledge-consistency alignment.
$\Phi$ is designed as:
\begin{equation}
\vspace{-0.1cm}
	\label{eq:phi}
	\Phi(\mathbf{h}_C, \mathbf{h}_\Delta) = \mathcal{M}(\mathbf{h}_C  \mathbin\Vert \mathbf{h}_\Delta),
    \vspace{-0.1cm}
\end{equation}
where  \( \mathbin\Vert \) denotes the concatenation operation, and \( \mathcal{M} \) is a learnable MLP block that operates on the combined vector.

Overall, we design the attestation framework by jointly optimizing the parameters $\theta_\Psi$ of knowledge vectorization network \(\Psi\) and $\theta_\Phi$ of knowledge fusion network \(\Phi\) to construct a metric space that preserves knowledge consistency between parent and child models. In this space, a higher similarity score indicates greater confidence that the fused knowledge \(\Phi(\mathbf{h}_C,\mathbf{h}_\Delta)\) aligns with the parent model’s knowledge \(\mathbf{h}_P\).

Specifically, true parent–child pairs and non-direct pairs are used as positive and negative samples, respectively, to learn the optimal parameters $\theta_\Phi^{*}$ and $\theta_\Psi^{*}$ as: 
\begin{equation}
	\begin{aligned}
		(\theta^*_\Phi, \theta^*_\Psi) &= \arg\min_{\theta_\Phi,\theta_\Psi} \mathcal{L}, \\
		\mathcal{L} =
		\,\texttt{sim}(\mathbf{h}_P, \Phi(\mathbf{h}_C, \mathbf{h}_\Delta))&-\,\texttt{sim}(\mathbf{h}_P, \Phi(\mathbf{h}'_C, \mathbf{h}'_\Delta))+m,
	\end{aligned}
\end{equation}
where \(\mathbf{h}'_C\) and \(\mathbf{h}'_\Delta\) denote the embeddings of a negative child model and its corresponding negative evolutionary direction, respectively, and \(\texttt{sim}(\cdot,\cdot)\) denotes cosine similarity. The margin \(m=0.2\) imposes a minimum distance between positive and negative pairs, thereby encouraging a more discriminative knowledge alignment space.

Finally, given an unknown parent–child model pair \((f_P,f_C)\), the knowledge consistency score is calculated as:
\begin{equation}
\vspace{-0.1cm}
	\label{eq53}
	\mathcal{S}(f_P,f_C) = \texttt{sim}(\mathbf{h}_P, \Phi(\mathbf{h}_C, \mathbf{h}_\Delta)).
    \vspace{-0.1cm}
\end{equation}
Then the lineage can be claimed
if and only if $\mathcal{S} (f_P,f_C)$ is above a pre-defined threshold $T$.

\textbf{Analysis of Cross-Generational Knowledge Alignment.} 
Models that undergo multiple successive fine-tuning stages often exhibit across-generational lineage (e.g., grandparent–grandchild), which can also be assessed via Eq.~(\ref{eq53}). However, as fine-tuning progresses, knowledge inherited from the parent model gradually diminishes, making lineage attestation increasingly difficult—especially in skip-generation cases where semantic alignment weakens. 

Consider a model \( f_P \) fine-tuned to \( f_{C1} \), and then to \( f_{C2} \). Let \( f_{\Delta} = \theta_1 + \theta_2 +\theta_0\) denote the evolution model obtained by task arithmetic, where \( \theta_1 \) and \( \theta_2 \) are task vectors of the two stages, respectively, $\theta_0$ is the initial parameters of $f_P$. As fine-tuning progresses across generations, the child model \( f_{C2} \) inherits progressively less knowledge from \( f_P \), and more of the original information is changed, which is captured by the task-vector corresponding to the evolution model \( f_\Delta \). Therefore, the information loss caused by the semantic inconsistency between the evolution model and the parent model, as well as the alignment bias introduced by the mismatch between model editing and the true optimization trajectory, accumulates progressively across generations.

This parameter-level accumulation of task vectors implies that the fusion knowledge vector  \( \Phi(\mathbf{h}_C, \mathbf{h}_\Delta) \) 
increasingly deviates from the parent model's original knowledge vector \( \mathbf{h}_P \). 
The degree of this semantic drift is positively correlated with the number of fine-tuning generations.
As a result, the similarity score between \( \mathbf{h}_P \) and \( \Phi(\mathbf{h}_C, \mathbf{h}_\Delta) \) for skip-generation relationships is expected to gradually decrease as the generational distance increases.

\begin{figure*}
	\begin{minipage}[t]{.97\linewidth}
		\centering
		\includegraphics[width=.99\linewidth]{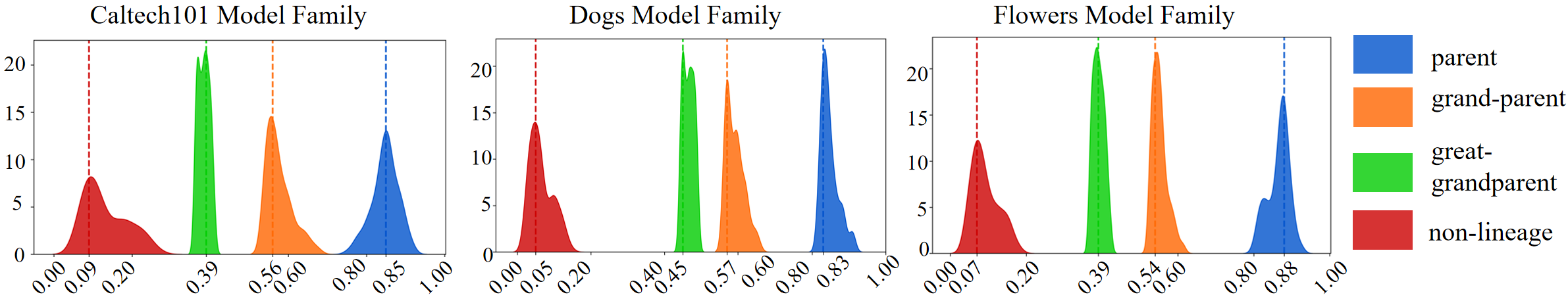}
	\end{minipage} 
	\caption{Kernel density estimation(KDE)
		of similarity scores for different lineage relations.	Each subplot shows the similarity distribution between a child model and its parent, grandparent, great-grandparent and non-lineage items, respectively. True parent-child pairs exhibit distinctly higher similarity, enabling effective rejection of false lineage claims.
		}
	\label{falseclaimkde}
    \vspace{-0.2cm}
\end{figure*}

\textbf{How to determine the Decision Threshold $T$?}
To explore an appropriate decision threshold \( T \) of lineage closeness, we construct multiple four-generation model families based on the ResNet-18 ~\cite{res18}, MobileNet-V2~\cite{mobilenet}, Diffusion-V2~\cite{dreambooth}, and LLaMA-3~\cite{llama3} architectures. 
Fig.~\ref{falseclaimkde} shows the kernel density estimation (KDE) of knowledge alignment similarity scores across different lineage relations of ResNet-18 families, and the other
results can be found in Fig.~\ref{diedaiarchi} of Appendix~\ref{ftroundsandarchi}. Each subplot originating from parent models trained on  Caltech101~\cite{caltech101}, Dogs~\cite{standogs}, and Flowers~\cite{flowers} dataset, respectively.

In all cases, the parent-child pairs (blue) consistently exhibit the highest similarity scores, with peak values centered 0.83-0.88, indicating strong knowledge consistency between adjacent generations. The grandparent relationships follow, peaking between 0.54–0.57, while the great-grandparent class peaking in the 0.39–0.45 range. In contrast, non-lineage pairs display the lowest similarity, concentrated near 0.00–0.20. This decay across generational distances demonstrates a clear and stable correlation between the embedding vector alignment and the true knowledge evolution, which is consistent with our previous analysis. Additionally, these results demonstrate that our MLA is capable of differentiating at least four generations of models.

Interestingly, the similarity intervals for lineage-nonlineage are well separated without significant overlap. This allows us to define tiered thresholds for lineage inference, enabling reliable differentiation between direct lineage, distant lineage, and non-lineage model pairs. 
For example, in model families derived from ResNet-18, a lower threshold (e.g., 0.3) allows reliable separation of three-generation lineage from non-lineage models, while a higher threshold (e.g., 0.7) precisely captures direct parent-child relationship. We further note that the exact similarity values slightly vary across different model architectures.
In the MobileNet-V2 model family, non-lineage pairs typically have similarity scores below 0.1, while grand-parent relations exceed 0.2; in contrast, the corresponding values for the Diffusion-V2 model family are approximately 0.2 and 0.4, respectively. 
Integrating the experimental findings from different families reveals that non-lineage similarities are mostly below 0.2, whereas parent--child similarities are generally above 0.7. Thus, a universal $T$ in this paper is used across different model families for coarse-grained attestation —potentially reusing the value from a known family with similar architecture. For a higher accuracy, a family-specific $T$ may be adopted by training in advance with labeled lineage pairs.



\section{Experiments}
\subsection{Experimental Settings}
\textbf{Basic Setup and Implementation Details.} We utilized the Adam\cite{adam} optimizer with a learning rate of $1\times10^{-4}$ for training the proposed framework, and all experiments were conducted on two NVIDIA A100 GPUs. For the classification-oriented knowledge encoder \( \Psi \),  the MLP block \( \mathcal{M} \) consists of a fully connected layer followed by a LayerNorm operation, which aims to project the input feature embeddings to a uniform 128-dimensional vector and ensure stable training. Next, the transformer encoder \( \mathcal{T} \) with 2 layers and a hidden size of 128 was applied. To handle varying numbers of categories across models, the input sequences were padded to the maximum number of categories in the candidate model using zero vectors.
Additionally, the knowledge fusion network \( \Phi \) was a lightweight projection block consisting of a single linear layer with dimensions (\(256\times128\)) followed by a ReLU activation.

For the diffusion-oriented knowledge encoder \( \Psi \), the input consists of a $256 \times 256$ image and its corresponding textual prompt. The feature extracted from the \texttt{UpBlock} layer has dimensions of $(320, 32, 32)$. The convolution layer $\mathcal{C}$ with a kernel size of $3 \times 3$ was then applied to this feature, followed by a global pooling operation $\mathcal{G}$ that transformed it into a $320$-dimensional vector. Finally, the MLP $\mathcal{M}$ was used to map this vector into a 160-dimensional knowledge embedding.

For the LLM-oriented knowledge encoder \( \Psi \), input features were first projected to a 512-dimensional vector space through a fully connected layer, followed by a LayerNorm layer. A second position-aware fully connected layer was then applied to capture local contextual patterns, followed by a dropout layer with a rate of 0.1. Finally, an adaptive average pooling layer aggregated the sequence into a 512-dimensional knowledge vector. Accordingly, the fully connected layer in $\mathcal{M}$ of the knowledge fusion network $\Phi$ was configured as 1024×512 dimensions.


\textbf{Classification Model Library}. For the ResNet-18\cite{res18} model families, we developed more than 500 models in total. Specifically, we trained 20 parent models on each of the Flowers\cite{flowers}, Dogs\cite{standogs}, and Caltech-101\cite{caltech101} datasets. 
To promote model diversity, we randomly partitioned classes for parent and child models and trained them with varied hyper-parameters.
The resulting (parent, child) model pairs were then divided, with 80\% used for training and 20\% for testing.
For across-generational relationship, we took the above mentioned parent models as Generation-1 and the child models as Generation-2. We then randomly fine-tuned them down two generations along Caltech101\cite{caltech101}→ Caltech101\cite{caltech101}→ CIFAR100\cite{cifar100}→ Aircraft\cite{aircraft}, Dogs\cite{standogs}→ Dogs\cite{standogs}→ CIFAR100\cite{cifar100}→ Oxford-IIIT Pet\cite{pet}, and Flowers\cite{flowers}→ Flowers \cite{flowers}→ CIFAR100 \cite{cifar100}→ Caltech101 \cite{caltech101}. The fourth-generation models were  primarily included to measure lineage closeness across generations. We also constructed a model family based on  MobileNet-V2 architecture~\cite{mobilenet}, as detailed described in Appendix~\ref{detailsetup}.

\textbf{Diffusion Model Library.}
We selected Stable-Diffusion-v2 version as the model family structure. From the  Hugging Face~\cite{transformers}
repository, we collected 30 distinct versions of \textit{Stable-Diffusion-v2} as the set of parent models. Subsequently, we apply two mainstream fine-tuning techniques—DreamBooth~\cite{dreambooth} and LoRA~\cite{lora}—to generate 100 fine-tuned variants for each method, which serve as the corresponding child models. To obtain knowledge feature representations for lineage analysis, we 
randomly sampled 1000 images from COCO2014~\cite{coco} dataset  as the probe sample set. Further details regarding the implementation of DreamBooth and LoRA are provided in the Appendix~\ref{detailsetup}.

\textbf{LLM Library}. We collected 5 derivative models and 22 additional fine-tuned variants of \textit{meta-llama/Llama-3.1-8B}, as well as 2 derivative models and 49 fine-tuned variants of \textit{Qwen/Qwen2.5-1.5B} from Hugging Face\footnotemark[1]. For LLaMA-families, we specify the parameters of  \textit{meta-llama/Llama-3.1-8B} as the initialization parameters $\theta_0$ and considered the original first-generation fine-tuned models, such as  \textit{TwinLlama-3.1-8B}, as parent models, while their fine-tuned descendants(such as  \textit{TwinLlama-3.1-8B-DPO}) served as child models. This setup was adopted because the initialization parameters of  \textit{meta-llama/Llama-3.1-8B}—which serve as the original parameters for training the parent model—were not available.  Based on descriptions in Hugging Face, we constructed probing datasets by sampling 50 prompts from each of six datasets and collecting five responses per prompt to capture model knowledge.
 Further details regarding the implementation of LLM libraries are provided in the Appendix~\ref{detailsetup}.

\begin{table*}[htbp]
	\caption{The performance of model lineage attestation(MLA) against across-generational attack(AGA) and weight pruning
		attack(WPA). N.L.refers to neural lineage \cite{Lineage}, N.P. refers to neural phylogeny \cite{NP} and T.I. refers to TinyImageNet dataset. }
	 \vspace{-0.2cm}
	\label{tab:classification}
    \centering
       \begin{adjustbox}{max width =0.92\textwidth}
	\begin{tabular}{c|cc|cc|cc|cc|cc|cc}
		\toprule
		\multicolumn{12}{c}{ResNet-18 Model Family}\\
		\hline
		\multirow{3}{*}{Methods}
		&\multicolumn{4}{c|}{Caltech101->Caltech101->CIFAR100}& \multicolumn{4}{c|}{Dogs->Dogs->CIFAR100}&\multicolumn{4}{c}{Flowers->Flowers->CIFAR100}
		\\
		&\multicolumn{2}{c|}{TPR}& \multicolumn{2}{c|}{FPR}&\multicolumn{2}{c|}{TPR}&\multicolumn{2}{c|}{FPR}&\multicolumn{2}{c|}{TPR}& \multicolumn{2}{c}{FPR}\\ 
		& AGA  & WPA & AGA  & WPA   & AGA  & WPA  & AGA  & WPA  & AGA  & WPA  & AGA  & WPA \\
		\midrule
		N.L. \cite{Lineage}  & 0.99  & 0.85  & 0.02 & 0.08 & 0.94 & 0.83 &0.04 &0.06 &0.95 &0.76&0.05&0.12 \\
		N.P. \cite{NP}  &0.58  &  0.43 & 0.21 & 0.33  &0.42 & 0.27 & 0.07 & 0.17 & 0.60&0.40&0.12&0.15 \\
		KRM \cite{KRM} & 0.19& 0.15  & 0.16 & 0.11  & 0.15  &0.14 & 0.13 & 0.13 & 0.20&0.19&0.10&0.13 \\
		IPGuard \cite{IPGuard} & 0.19 & 0.16  & 0.15 & 0.14  & 0.16 & 0.14 & 0.14 & 0.12&0.20 & 0.13&0.11&0.12 \\
		Ours & \textbf{1.00}&\textbf{0.99}&\textbf{0.00}&\textbf{0.01}& \textbf{0.99}&\textbf{0.99}&\textbf{0.02}&\textbf{0.03}&\textbf{0.99}&\textbf{0.99}&\textbf{0.01}&\textbf{0.00}\\
		\hline
		\multicolumn{12}{c}{MobileNet-V2 Model Family}\\
		\hline
		\multirow{3}{*}{Methods}
		
		&\multicolumn{4}{c|}{Caltech101->T.I.->Mixed Dataset}& \multicolumn{4}{c|}{T.I.->Pet->Mixed Dataset}&\multicolumn{4}{c}{Flowers->CIFAR100->Mixed Dataset}
		\\
		
		&\multicolumn{2}{c|}{TPR}& \multicolumn{2}{c|}{FPR}&\multicolumn{2}{c|}{TPR}&\multicolumn{2}{c|}{FPR}&\multicolumn{2}{c|}{TPR}& \multicolumn{2}{c}{FPR}\\ 
		& AGA  & WPA & AGA  & WPA   & AGA  & WPA  & AGA  & WPA  & AGA  & WPA  & AGA  & WPA \\
		\hline
		N.L. \cite{Lineage} &1.00&0.83&0.02&0.08&0.95&0.92&0.13&0.16&0.96&0.94&0.10&0.13 \\
		N.P. \cite{NP} &0.73  & 0.66 &0.11& 0.08&0.65&0.40&0.15& 0.20& 0.63& 0.45& 0.17&0.15 \\
		KRM \cite{KRM}  & 0.18&0.15&0.18&0.11&0.22&0.14&0.18&0.13&0.19&0.18&0.18&0.13 \\
		IPGuard \cite{IPGuard} &0.18&0.16&0.13&0.14&0.14&0.14&0.12&0.20&0.13&0.11&0.12&0.11 \\
		Ours &\textbf{1.00}&\textbf{0.99}&\textbf{0.00}&\textbf{0.03}&\textbf{0.99}&\textbf{0.99}&\textbf{0.01}&\textbf{0.03}&\textbf{0.99}&\textbf{0.98}&\textbf{0.01}& \textbf{0.01}\\
		\hline
	\end{tabular}
    \end{adjustbox}
	\label{testclass}
    \vspace{-0.2cm}
\end{table*}

\textbf{Baselines}. We evaluate our proposed MLA framework against four typical baselines: the state-of-the-art methods Neural Lineage\cite{Lineage} (Abbreviated as N.L.), Neural Phylogeny\cite{NP} (Abbreviated as N.P.), KRM~\cite{KRM}, and IPGuard~\cite{IPGuard}.  
Since the number of classes in parent–child pairs may differ, and both KRM and IPGuard rely on class-specific boundary samples, we use a small number of classes in each pair for consistency.
The N.L. method focuses on the coupling between model weights and features extracted from the fine-tuning dataset. As such access is infeasible under our attack setting, we instead use the parent model’s training data to extract the aligned features for the N.L. method, which preserves the core feature–weight relationship needed for evaluation.

It should be noted that KRM and IPGuard are designed specifically for classification models, and thus they are not applicable in our MLA experiments for diffusion model and LLM, where we compare only with the N.L. and N.P. methods.     \textit{Furthermore, the N.L. method requires storing a large number of feature representations and model parameters, making it infeasible to run on the stable diffusion model and LLM experiments}.

\textbf{Evaluation Metrics.}  
We adopt the following metrics to evaluate the performance of the lineage attestation methods, i.e., the ability to fairly determine the presence or absence of a lineage relationship between attack models and victim models under various threat scenarios. (1)True Positive Rate (TPR): The proportion of positive samples (i.e., model pairs with a true lineage relationship) that are correctly identified.
(2) False Positive Rate (FPR): The proportion of negative samples (i.e., model pairs without a lineage relationship) that are incorrectly classified as positive.
(3) Receiver Operating Characteristic (ROC) curve:It plots the TPR against the FPR at various decision thresholds.

In the \textit{Illegal Model Derivative Verification} scenario, the lineage attestation mechanism  seeks to determine whether a suspect model $f_C^\mathcal{A}$ is directly derived from a protected model \( f_P \). The task is to distinguish true lineage pairs from unrelated ones, where a high TPR indicates effective detection and a low FPR reflects strong robustness against false accusations.

In the \textit{False Lineage Claim Reject} scenario, the verifier must distinguish true direct lineage from indirect or non-lineage relations to invalidate false claims. The goal of this task is to discriminate positive pairs \((\text{parent}, \text{child})\) from negative pairs  \((\text{non-direct parent}, \text{child})\), with the ROC curves used to evaluate the identification performance.

\subsection{MLA Performance for Illegal Model Derivative Verification} 
\label{sec:illegal}
We conduct experiments under two attack scenarios: the Across-Generational Attack (AGA) and  Weight Pruning before Fine-tuning Attack (WPA).  For AGA, we evaluate MLA method’s capability to detect lineage relationships within three generations. For WPA, we adopt a pruning rate $p$ of 10\% following the KRM~\cite{KRM} setup.

\textbf{MLA across Classification Model Families.} 
In our approach, the decision threshold $T$ for lineage detection is set to 0.3. Table \ref{tab:classification} shows the lineage attestation results of both ResNet-18 model family and MobileNet-V2 model family.

First of all, it can be observed 
that our method consistently achieves near-perfect performance against AGA across both ResNet-18\cite{res18} and MobileNet-V2\cite{mobilenet} architectures, with TPR around 0.99 and FPR not exceeding 0.02. 
This superior performance stems from our design, which captures the consistent trajectory of both model parameters alternation and knowledge evolution during fine-tuning, enabling robust adaptation to diverse fine-tuning strategies.
In contrast, the N.L. \cite{Lineage} method achieved suboptimal performance. It attempts to encode the evolution of a model by examining the coupling between its parameters and the extracted data features. However, this method fails to precisely characterize the nuanced knowledge transfer process that occurs during fine-tuning. The N.P.~\cite{NP} method suffers even more severe degradation, with TPR dropping to 0.42 and FPR rising to 0.21 under the  worst-case, due to its reliance on parameter-level differences while overlooking the fine-tuning-induced semantic variations among models. Comparably, both KRM\cite{KRM} and IPGuard\cite{IPGuard} perform poorly, with TPRs ranging from 0.10 to 0.20 and relatively high FPRs. These methods assume stable decision boundaries during fine-tuning, that may hold under lightweight fine-tuning adaptations but fail under substantial task shifts, where models significantly reshape their boundaries to fit new fine-tuning task.

Furthermore, we observe that our method remains effective under WPA scenario, as the pruned models still retain sufficient knowledge to support reliable attestation.  In comparison, N.L.~\cite{Lineage}, which directly encodes parameter changes, shows a noticeable drop in performance under pruning. This is particularly obvious in the Caltech101-based model family, where it suffers a TPR drop of 0.14 in ResNet-18 model family and 0.17 in MobileNet-V2 model family, highlighting its sensitivity to weight perturbations introduced by pruning. Similarly, these perturbations also impact N.P.~\cite{NP}, reducing its detection capability with TPR dropping as low as 0.27 in the worst case. WPA has limited impact on KRM\cite{KRM} and IPGuard\cite{IPGuard}, since it does not significantly alter the model's decision boundaries—on which these methods heavily rely. However, despite their stability, their overall TPR remains low (typically between 0.10 and 0.20), indicating limited capability in accurately detecting true lineage.

\textbf{MLA across Diffusion Model Families.}
For diffusion models, \textit{DreamBooth} and \textit{LoRA} are two widely used fine-tuning techniques, and we evaluate whether our method and N.P. can distinguish models fine-tuned using these techniques from unrelated models.
As shown in Table~\ref{tab:diff}, our method shows strong robustness to AGA and WPA across both  DreamBooth and LoRA approaches. 
We attribute this robustness to the fact that these fine-tune approaches enable diffusion models to adapt to specific tasks by embedding new knowledge, while our MLA method effectively captures this knowledge-level alignment.
In contrast, the baseline method N.P. relies on parameter-level clustering and lacks a mechanism to capture knowledge evolution. As a result, it exhibits  a marginal drop in TPR against AGA (0.09 lower than ours in LoRA),  and suffers a more noticeable decline under WPA than under AGA, with a 0.04 drop on both DreamBooth and LoRA. The attestation signal reflecting model lineage becomes increasingly obscured under WPA. Additionally, N.P. consistently shows high FPR (0.20–0.23), likely due to its clustering-based strategy, which fixes cluster boundaries based on training models. This may lead to test models being misclassified into incorrect clusters, especially when different parent models exhibit similar parameter patterns.
\begin{table}[htbp]
	\centering
	\caption{MLA performance on Stable Diffusion models.}
	   \vspace{-0.1cm}
	\label{tab:diff}
    \begin{adjustbox}{max width =0.47\textwidth}
	\begin{tabular}{lcccccc}
		\toprule
		\multirow{2}{*}{}
		& \multicolumn{2}{c}{Dreambooth (TPR)} 
		& \multicolumn{2}{c}{LoRA (TPR)}
		& \multicolumn{2}{c}{FPR}\\
		\cmidrule(lr){2-3} \cmidrule(lr){4-5} \cmidrule(lr){6-7}
		& AGA & WPA & AGA & WPA & AGA & WPA \\
		\midrule
		N.P.\cite{NP}  & 0.94 & 0.90 & 0.91 & 0.87 & 0.23 & 0.20 \\
		Ours  & \textbf{1.00} & \textbf{0.99} & \textbf{1.00} & \textbf{0.99} & \textbf{0.02} & \textbf{0.02} \\
		\bottomrule
	\end{tabular}
    \end{adjustbox}
\end{table}
 
\begin{figure*}
	\begin{minipage}[b]{.24\linewidth}
		\centering
		\includegraphics[width=.99\linewidth]{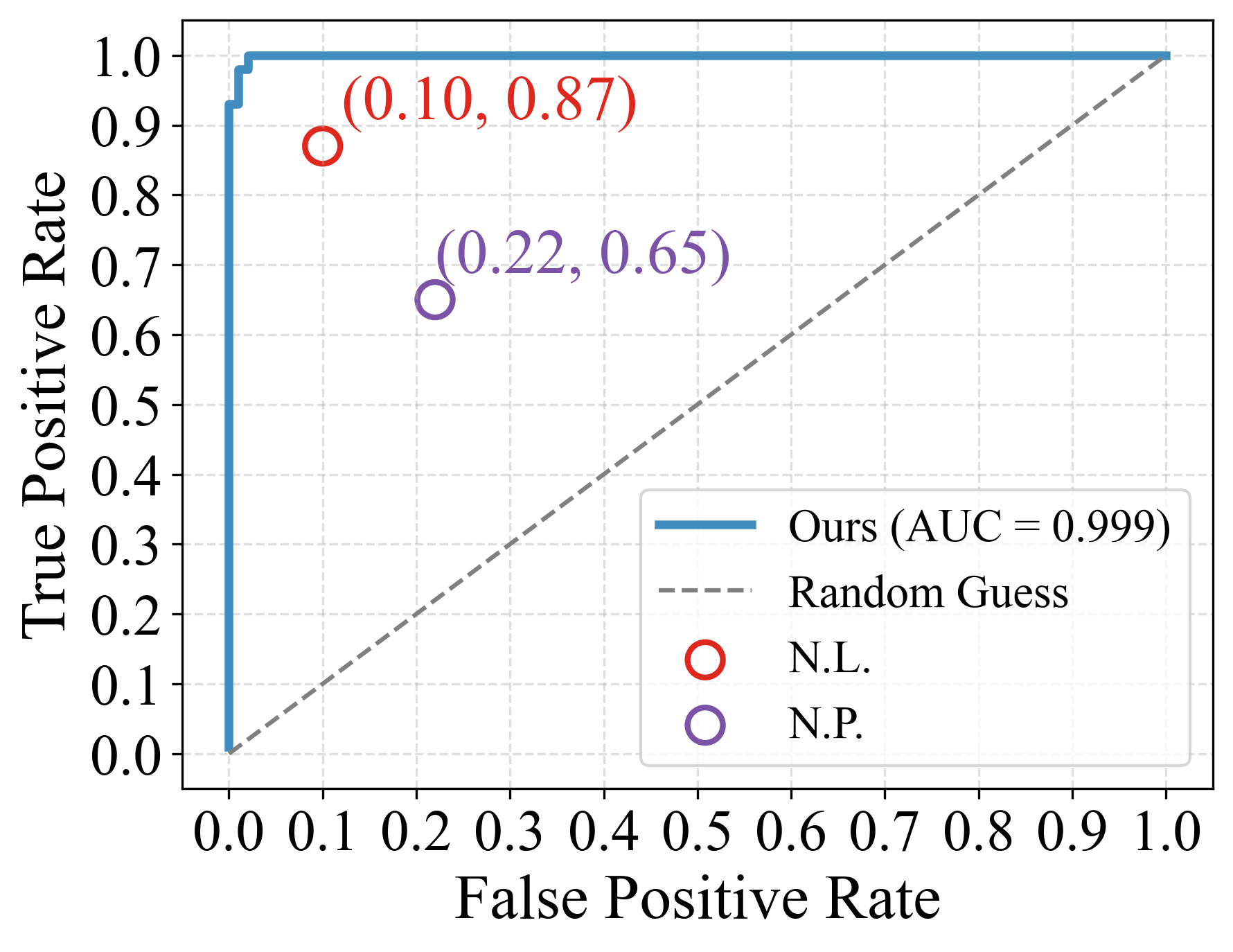}
		\caption*{(a)~ResNet-18 family.}
	\end{minipage} 
	\medskip
	\begin{minipage}[b]{.24\linewidth}
		\centering
		\includegraphics[width=0.99\linewidth]{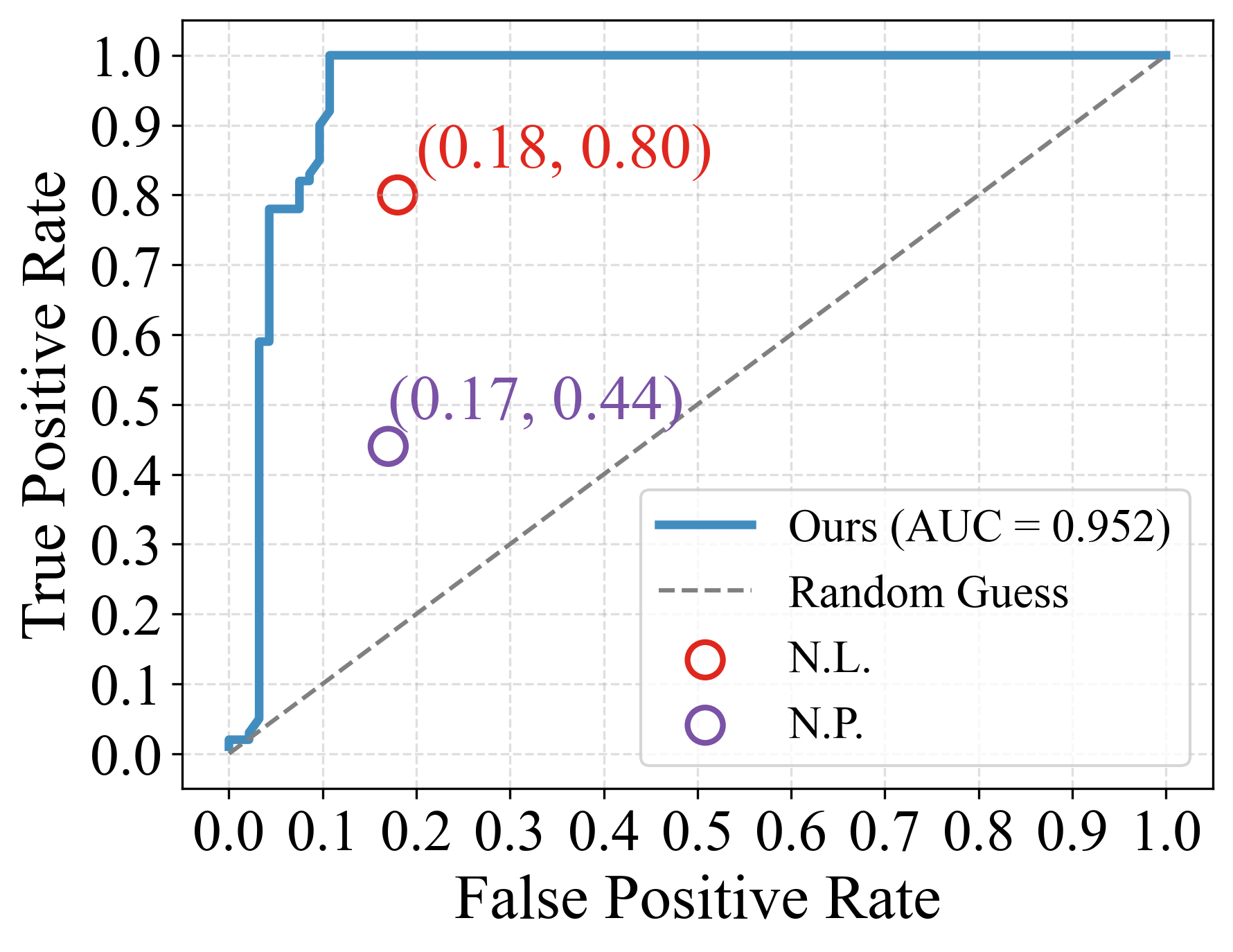}
		\caption*{(b)~MobileNet-V2 family.}
	\end{minipage}
	\begin{minipage}[b]{.24\linewidth}
		\label{rocdiffussion}
		\centering
		\includegraphics[width=0.99\linewidth]{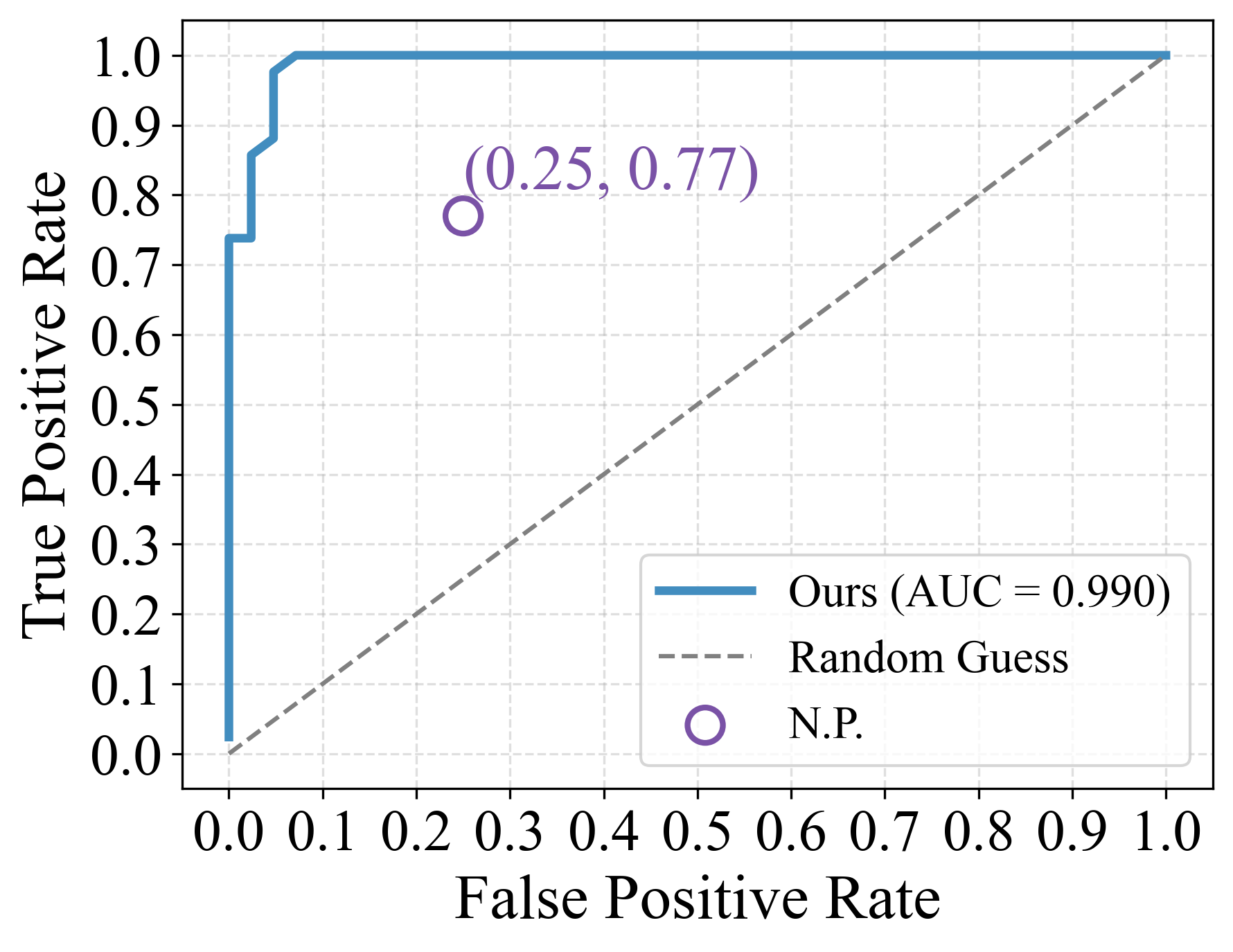}
		\caption*{(c)~Stable Diffusion family.}
	\end{minipage}
	\begin{minipage}[b]{.24\linewidth}
		\label{rocllama3}
		\centering
		\includegraphics[width=0.99\linewidth]{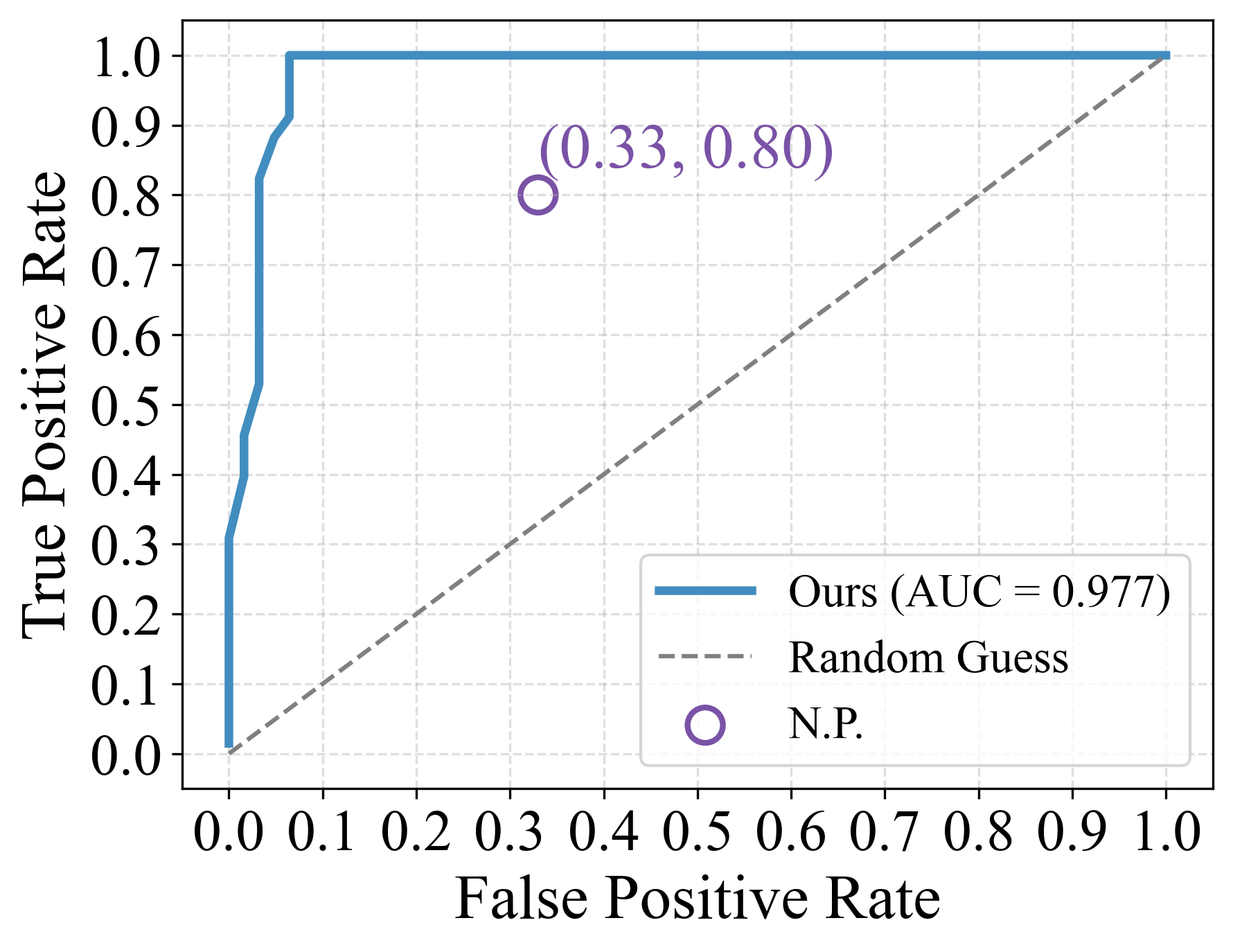}
		\caption*{(d)~LLaMA3-family.}
	\end{minipage}
	\vspace{-0.15cm}
	\caption{ROC of false lineage claim rejection  for various types of model families.}
	\label{fig:roc}
\end{figure*}
\textbf{MLA across LLMs Families.}
For the LLM library, the models are realistic ones directly collected from the open-source repository, not trained by us. Thus we can't perform the WPA experiments and have only the results for AGA.
As shown in Table \ref{tabllm}, our method effectively identifies model lineage relationships, outperforming the baseline (N.P.).  
Despite potential semantic discrepancies between the probe datasets and the parent models, the experimental results show that our method remains robust. 
This improvement can be attributed to its ability to capture fine-grained knowledge transfer patterns during fine-tuning, whereas the N.P. method, relying on heuristic and less adaptive features, is more susceptible to dynamically semantic inconsistencies among models.
\begin{table}[htbp]
	\caption{MLA performance on LLMs.}
	 \vspace{-0.1cm}
	\label{tabllm}
        \centering
        \begin{adjustbox}{max width =0.42\textwidth}
	\begin{tabular}{ccccc}
		\toprule
		\multirow{2}{*}{Model}&\multicolumn{2}{c}{Ours}&\multicolumn{2}{c}
		{N.P.\cite{NP}} \\ \cmidrule(lr){2-3}  \cmidrule(lr){4-5}
		&TPR &FPR&TPR &FPR \\
		\hline
		Llama3.1-8B-family&0.97 &0.02& 0.80&0.33
		\\Qwen2.5-1.5B-family&0.99 &0.01& 0.79&0.21\\
		\bottomrule
	\end{tabular}
    \end{adjustbox}
	\label{testclass} 
\end{table}

\subsection{MLA Performance for False Lineage Claim Rejection}
\label{sec:falseclaim}
Fig.~\ref{fig:roc} shows the ROC curves comparing our method with two binary-based decision baselines, N.P.\cite{NP} and N.L.~\cite{Lineage}, in the lineage verification setting. Unlike our approach, which provides a continuous confidence score for flexible thresholding, the baselines yield only binary decisions, thus appearing as discrete TPR/FPR points. As a result, their performance is represented as the discrete TPR/FPR points.

As shown in Fig.~\ref{fig:roc}, our method achieves an AUC (area under the curve) of 0.999 on the ResNet-18 family and 0.952 on the MobileNet-V2 family, demonstrating strong discriminative ability across a wide range of thresholds. N.L. achieves a lower FPR (0.10) and a comparable TPR (0.87) on ResNet-18, while suffers on MobileNet-V2 with a TPR of 0.80 and a much higher FPR of 0.18. This degradation may stem from the fact that its decision relies on behavioral similarity in weight or output features, which can also be shared by non-direct lineage models within the same family.
N.P. achieves a TPR of 0.65 but incurs a high FPR of 0.22, indicating a strong tendency to misidentify unrelated models as related. This is likely due to its data-free design, which relies solely on parameter-level similarity. As models from the same family often converge along similar optimization trajectories, N.P. fails to effectively distinguish direct lineage from architectural resemblance.

We further evaluate the generalizability of our method on generative and large language model families, as shown in Figs.~\ref{fig:roc}(c) and (d). Our method achieves an AUC of 0.990 on Stable Diffusion-V2 and 0.977 on the LLaMA3 family, indicating strong and consistent performance across diverse model architectures and modalities. The superiority lies in its knowledge-centric design,  especially in large-scale models where richer knowledge enhances signal separability.
In contrast, existing baselines such as N.P and N.L. exhibit a marked inability to differentiate these relationships.
Relying solely on feature-level or parameter-level similarity is insufficient for effective alignment of intrinsic knowledge consistency. The inconsistent performance of N.P. across the four model architectures suggests that it is sensitive to both semantic and structural variations, reflecting a lack of robustness.

\subsection{Resistance to Adaptive Attacks}
In practice, adversaries may become aware of our MLA protocol and develop adaptive attacks to evade verification. They may attempt to disrupt knowledge-consistency verification by altering the model architecture, or to undermine knowledge-vector extraction by perturbing model parameters.

\textbf{Structure Evasion Attack(via distillation).} 
In this attack, the adversary first steals the protected model $f_P$ and fine-tunes it to obtain a descendant model $f_C^\mathcal{A}$. To evade lineage verification, the adversary then applies architecture transformation strategies such as knowledge distillation, transferring $f_C^\mathcal{A}$ into a structurally different student model $f_C^\mathcal{A'}$.
By design, $f_C^\mathcal{A'}$ compromises parameter  or structure-based consistency checks, rendering them ineffective.

Model lineage depicts the fine-tuning relationship between two models, emphasizing their knowledge evolutionary path in fine-tuning,
  while model distillation (born-again neural networks) resembles knowledge cloning, lacking both parameter inheritance and continuity in the knowledge evolution, and thus it can corrupt model lineage. 
However, the behavioral traits of the teacher model can still be inherited by the student model, leading to an atypical pattern of knowledge consistency between the two models.

To assess the knowledge consistency between the suspect model $f_C^\mathcal{A'}$ and the victim model $f_P$, we perform reverse distillation: starting from the initial parameters of $f_P$, we distill $f_C^\mathcal{A'}$ into a proxy model $f_C^\mathcal{A''}$ that has the same architecture with $f_P$ and then evaluate its lineage similarity with $f_P$. Thus, the lineage relationship between $f_C^\mathcal{A'}$ and $f_P$ can be measured by the knowledge consistency between $f_C^\mathcal{A''}$ and $f_P$.
For evaluating the robustness of our MLA to the structure evasion attack (via distillation), we simulate it by distilling the classification model
$f_C^\mathcal{A}$ into the compact network $f_C^\mathcal{A'}$ composed of three convolutional layers in a data-free manner using deep inversion~\cite{deepinv}. We then perform reverse distillation from $f_C^\mathcal{A'}$ towards $f_P$ to obtain $f_C^\mathcal{A''}$ and measure the resulting lineage similarity to verify knowledge inheritance.    


Table~\ref{abldis} reports the resistance of our MLA against distillation attacks under varying student model qualities. When the student model achieves a high test accuracy (65.43\%), corresponding to a well-distilled model, MLA attains a TPR of 0.76, indicating almost perfect detection of lineage relationships. As the distillation quality decreases (test accuracy drops to 48.62\%), the overlap of inherited knowledge diminishes, leading to a lower TPR of 0.36. Usually, the goal of a distillation attack is to preserve a high-quality (i.e. a higher test accuracy more than 65.43\%) student model. That said, the attacker is still difficult to bypass our MLA’s verification unless the student model is with very low accuracy.
\begin{table}[htbp]
	\caption{Resistance of MLA to Distillation Attacks. Acc.  refers to the classification accuracy.  }
    \centering
    \begin{adjustbox}{max width =0.46\textwidth}
	\begin{tabular}{c|c|c|c|c}
		\toprule
		Acc.(\%) &89.90($f_P$)&65.43($f_C^\mathcal{A'}$)&50.32($f_C^\mathcal{A'}$)&48.62($f_C^\mathcal{A'}$)\\
		\midrule
		TPR&0.99 &0.76 & 0.54 &0.36 \\
		\bottomrule
	\end{tabular}
	    \end{adjustbox}
    \label{abldis}
\end{table}


\begin{figure}[htbp]
	\centering
	\includegraphics[width=0.8\linewidth]{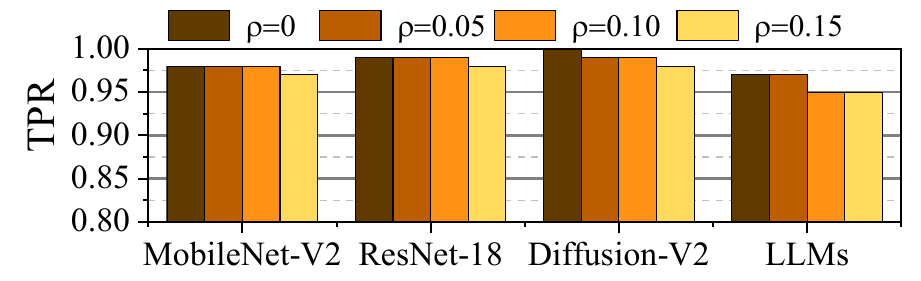}
	\caption{ Lineage attestation accuracy across different model families  against varying parameter perturbation radio $\rho$.}
	\label{perturfig}
    \vspace{-0.2cm}
\end{figure}

\textbf{Parameter Perturbation Attack.}
We simulate an adaptive attack that perturbs model parameters to hinder obtaining correct evolution-model and knowledge-vector. Following NPET~\cite{pertur}, we add
zero-mean uniform noise scaled to each parameter’s mean absolute value.
The perturbation ratios $\rho \in \{0.05, 0.1, 0.15\}$ specify the fraction of mean absolute values used to control the resulting standard deviation of the noise.
As shown in Fig.~\ref{perturfig}, these parameter perturbations do not noticeably
affect knowledge alignment, despite slightly altering the model weights without
impacting inference. This indicates that MLA is resilient to random parameters
noise, and such minor perturbations are insufficient to compromise lineage
verification.

\textcolor{purple}{\textbf{Stronger Weight-Pruning Attacks.}}
\textcolor{purple}{Weight-pruning serves as an adaptive attack, as removing neurons may disrupt the knowledge patterns used for lineage assessment. 
We evaluate the resilience of MLA to larger weight-pruning rates. Since weight-pruning attack inevitably degrades model performance that reflects the model’s practical usability, we assess the attack-induced model performance degradation using task-appropriate metrics. Specifically, for LLMs, we measure the relative increase ($\Delta$~PPL(\%)) in Perplexity (PPL)~\cite{ppl} compared to the original model; for diffusion models, we evaluate the relative change ($\Delta$~FID(\%)) in Fréchet Inception Distance (FID)~\cite{fid}; and for classification models, we report the percentage of accuracy loss ($\Delta$~Acc.(\%)).
PPL measures a language model’s uncertainty in predicting the next token given the preceding context and  an increase in PPL directly reflects a degradation in the model’s ability to generate fluent and coherent text.
FID, on the other hand, quantifies the discrepancy between the distributions of generated images and real images in the feature space, and an increase in FID indicates that the generative model approximates the target image distribution poorly.}
\begin{figure*}[t!]
    \centering
        \begin{minipage}{0.98\textwidth}
    \begin{minipage}{0.32\linewidth}
        \centering
        \includegraphics[width=0.85\linewidth]{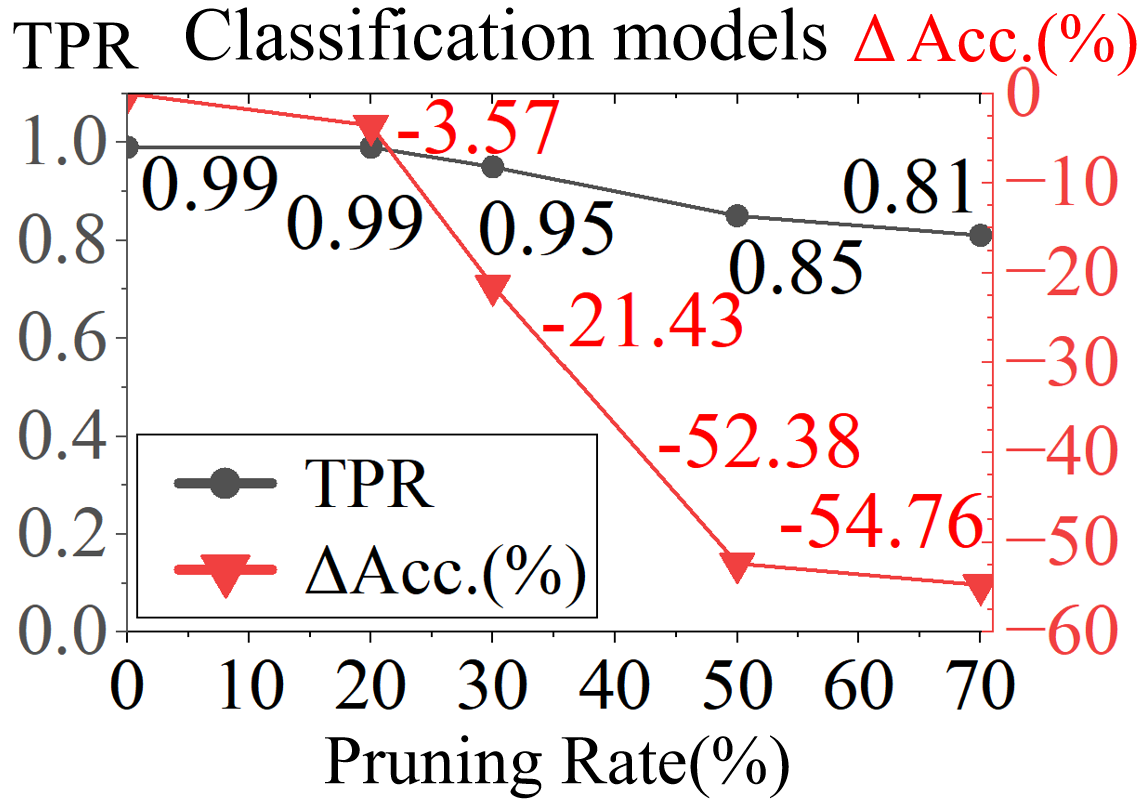}
        \caption*{(a) WPA for classification models.}
    \end{minipage}
    \hfill
    \begin{minipage}{0.32\linewidth}
        \centering
        \includegraphics[width=0.85\linewidth]{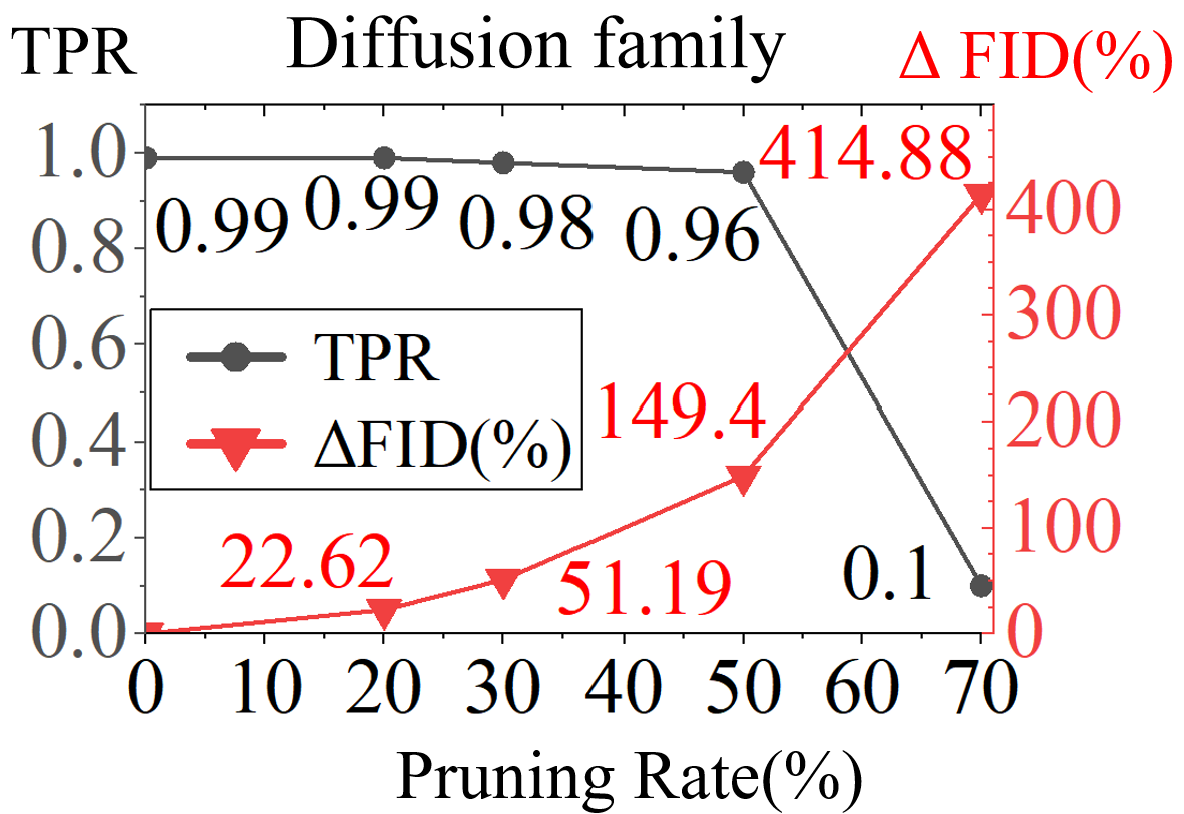}
        \caption*{(b) WPA for diffusion models.}
    \end{minipage}
    \begin{minipage}{0.32\linewidth}
        \centering
        \includegraphics[width=0.85\linewidth]{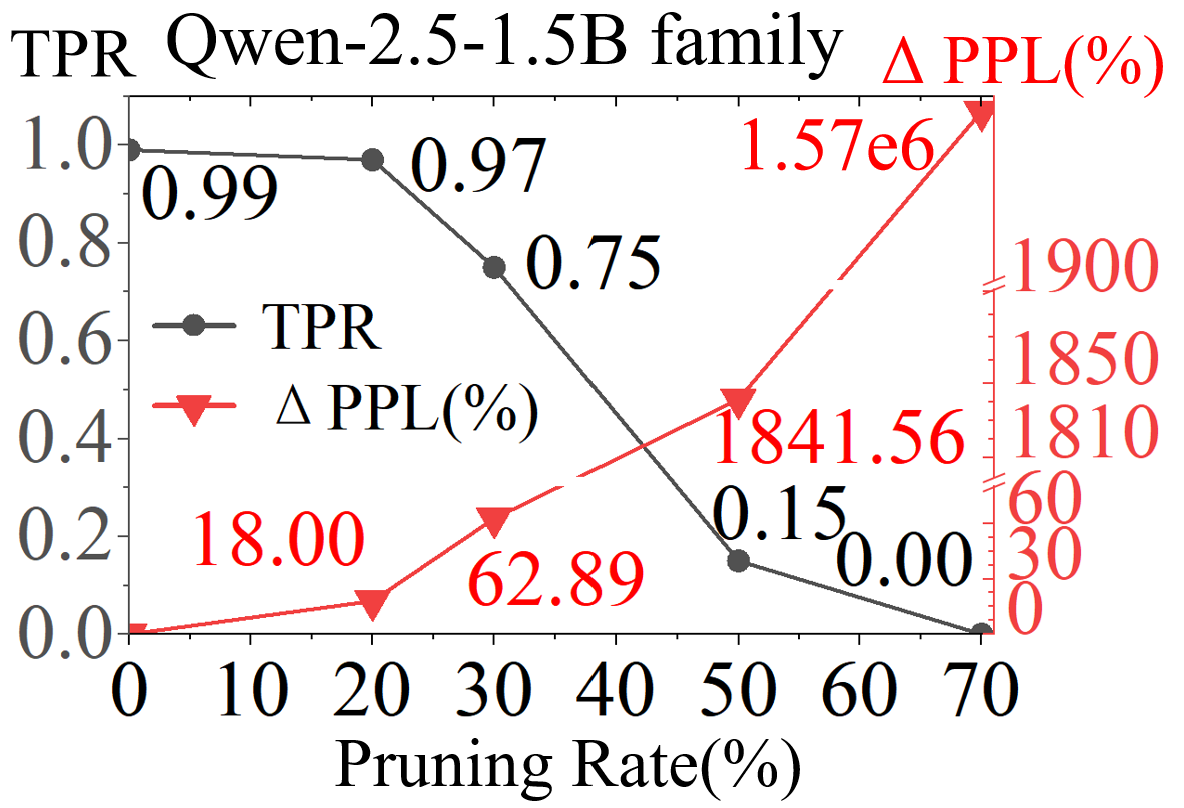}
        \caption*{(c) WPA for LLMs.}
    \end{minipage}
        \caption{\textcolor{purple}{MLA performance against Weight Pruning Attack.}}
        \label{figwpa}
        \end{minipage}
    \hfill
\end{figure*}

\textcolor{purple}{For classification models, child models are constructed via fine-tuning without introducing new classes, such that the number of target classes remains unchanged. As shown in Fig.~\ref{figwpa}(a), the TPR (averaged over the ResNet-18 and MobileNet-V2 families) decreases as the pruning ratio increases, accompanied by a corresponding drop in classification accuracy.
This behavior can be attributed to the fact that, when the class set is unchanged, the classification knowledge encoded in the decision boundaries largely preserves its fundamental structure even after pruning. As a result, classification models exhibit relatively higher robustness under increasing pruning strength.
In contrast, as shown in Fig.~\ref{figwpa}(b) and Fig.~\ref{figwpa}(c), LLMs and diffusion models demonstrate significantly higher vulnerability under stronger pruning attacks. As the pruning ratio increases, features that encode the model’s sensitivity to training data are progressively removed, leading to a substantial degradation in lineage attestation accuracy.
It is also worth noting that, although classification models maintain relatively high lineage identification performance under moderate pruning, their classification accuracy drops sharply once the pruning ratio exceeds 50\%, which severely limits their practical usability. A similar trend is observed for diffusion models and LLMs, where aggressive pruning (e.g., 50\%) results in substantial degradation of both generation quality and model usability.}

\textcolor{blue}{\textbf{Knowledge-Overwriting Attack.}
In the Illegal Model Derivative Verification setting, an adversary aware of the
probe dataset may attempt to corrupt the extracted knowledge by injecting behaviors
that conflict with probe-induced signals. We assume that the adversary possesses prior knowledge of the probe dataset 
(e.g., images for classification models, generative images and captions for diffusion models, and probe prompts
for LLMs). We vary the attack strength by perturbing different portions of the 
probe set 
for the following mechanisms:
(i) \textit{Label-Perturbed Fine-tuning for Classification Model.}
The adversary re-labels a subset of samples within each selected class, inducing shifts in the decision regions tied to probe data.
(ii) \textit{Caption-Corrupted Fine-tuning for Diffusion Models.}
The adversary substitutes probe captions with incorrect descriptions and fine-tunes adapters on the resulting mismatched image–caption pairs, disrupting probe-induced semantic alignment.
(iii) \textit{Cross-Family Response Fine-tuning for LLMs.}
The adversary replaces the correct probe outputs with responses from an unrelated
model (e.g., Llama-3.1-8B), using them as pseudo-targets during fine-tuning to
inject cross-family behavioral patterns.
As shown in Fig.~\ref{fig-overwriting},
the TPR values of classification models degrade only slightly because boundary perturbations do not alter their underlying decision structure.  
For diffusion models, caption corruption is localized to the mismatched image–caption pairs, leaving semantic knowledge within unrelated probe largely unaffected. In contrast, LLMs are more vulnerable since the cross-family pseudo-targets shift their conditional token distribution.   We observe that the knowledge overwriting attack can forcibly inject
conflicting behavioral signals into the suspect model. 
However, such
effects can be straightforwardly mitigated by evaluating the model on
an independently constructed probe dataset.
}
\begin{figure}[htbp]
        \centering
        \includegraphics[width=0.65\linewidth]{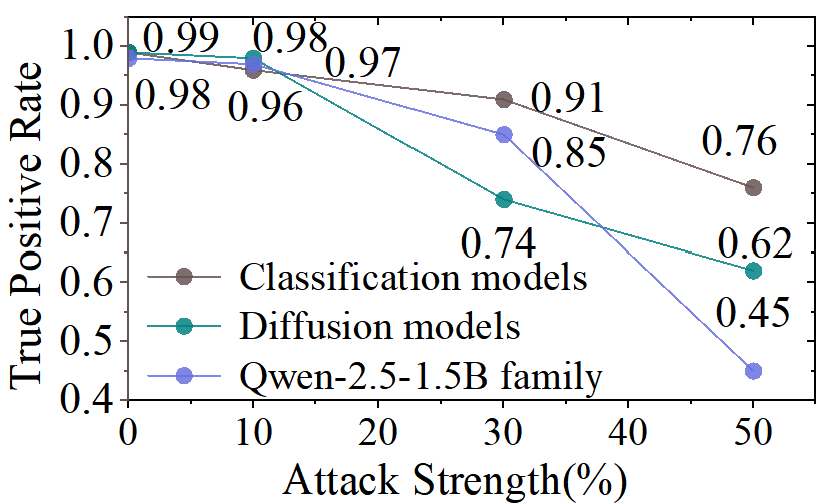}
        \caption{MLA performance against Knowledge-Overwriting Attack.}
            \label{fig-overwriting}
\end{figure}




\subsection{Ablation Study}
In this section, we present ablation studies to examine the influence of critical factors in our method, including the effect of probe data size and the generalization ability for different LLM families. Additional results are provided in Appendix~\ref{probealterdadta}.



\begin{figure}[htbp]
    \centering
    \begin{minipage}{0.47\linewidth}
        \centering
        \includegraphics[width=\linewidth]{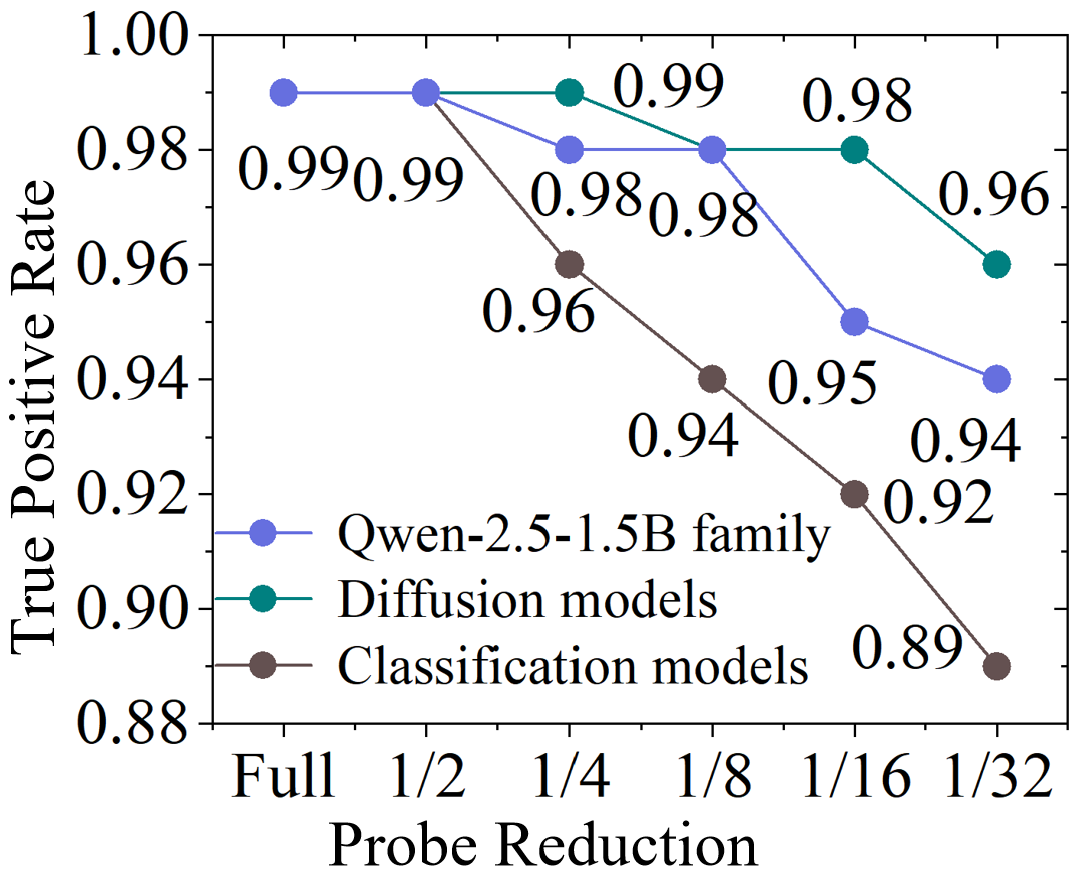}
        \caption*{(a) Varying data size. }
    \end{minipage}
    \hfill
    \begin{minipage}{0.47\linewidth}
        \centering
         \includegraphics[width=\linewidth]{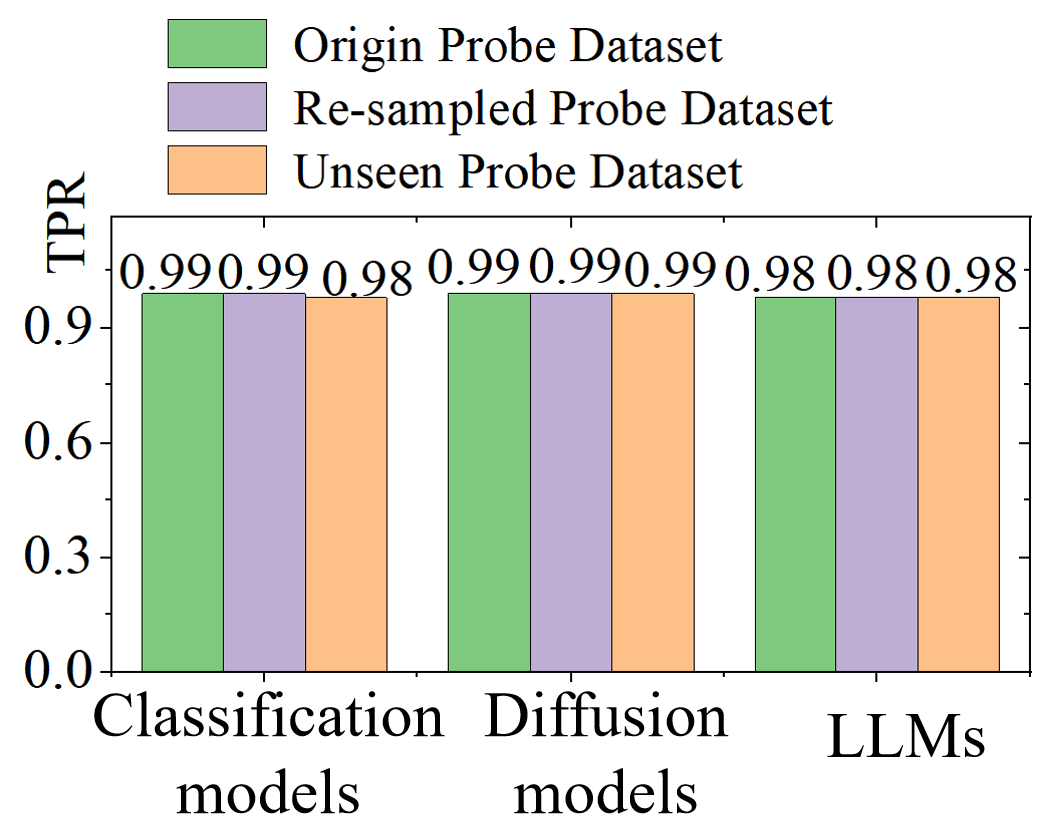}
        \caption*{(b) Varying data source.}
        \label{knowledgeoverwriting}
    \end{minipage}
    \caption{\textcolor{blue}{MLA performance against varying probe dataset.}}
	\label{probe}
    \vspace{-0.2cm}
\end{figure}

\textcolor{blue}{
\textbf{Effect of Varying Probe Dataset.}
To evaluate the sensitivity of MLA to probe quality, we evaluate its robustness with respect to both varying probe size and varying probe data source. The probe data are constructed either by re-sampling an equal number of samples from the same dataset, or by drawing samples from unseen datasets during training (detailed in Section~\ref{probealterdadta}). 
As shown in Fig.~\ref{probe}, classification models are only mildly affected by different data source. Furthermore, performance on diffusion models and LLMs remains largely stable under both different probe size and varying data source.
Even when the probe is drawn from an unseen dataset, diffusion models and LLMs still exhibit non-trivial semantic alignment on these data, owing to their strong priors and extensive pretraining on large-scale datasets.}

\textcolor{blue}{ 
\textbf{LLM Families Generalization.}
To further evaluate the generalizability of MLA beyond the model
families used in the main experiments, we test its performance on
20 additional HuggingFace variants spanning four major ecosystems:
Microsoft Phi-4, Alibaba Qwen-3, Mistral, and Google Gemma. The
complete list of evaluated model variants is provided in our OpenScience
repository.\footnote{\url{https://zenodo.org/records/17958570.}} As
shown in Table~\ref{tab:vendor-tpr-fpr}, MLA achieves consistently high
TPR (0.98–0.99) and low FPR (0.01–0.03) across all four families.  These results demonstrate that the proposed lineage attestation framework generalizes robustly across heterogeneous 
architectures and training pipelines, indicating its applicability to a broad range of open-weight LLM ecosystems.}

\begin{table}[h]
\centering
\setlength{\tabcolsep}{9pt}
\renewcommand{\arraystretch}{1.15}
\caption{\textcolor{blue}{MLA performance across additional LLM families.}}
\label{tab:vendor-tpr-fpr}
\begin{adjustbox}{max width =0.43\textwidth}
\begin{tabular}{lcccc}
\toprule
 & \textbf{Phi-4} & \textbf{Qwen-3} & \textbf{Mistral} & \textbf{Gemma} \\
\midrule
TPR & 0.99 & 0.98 & 0.99 & 0.98 \\
FPR & 0.01 & 0.01 & 0.03 & 0.01 \\
\bottomrule
\end{tabular}
\end{adjustbox}
\vspace{-0.4cm}
\end{table}

\section{Conclusion}
In this work, we unravel the knowledge mechanism of fine-tuning-induced model evolution. 
Based on this rationale, we address the model lineage attestation problem by exploring the knowledge evolution in neural networks and propose a model lineage attestation framework. In the framework, knowledge modifications by fine-tuning are explicitly captured through model editing, and further a knowledge vectorization technique is proposed to embed these edited changes into a common latent space. This enables robust alignment in terms of the consistency among the resulting knowledge embeddings, providing the fundamental evidence for confirming model lineage relationship. Extensive experimental evaluations under two key adversarial scenarios validate the effectiveness of our method. Our framework offers a robust solution for verifying the lineage of model derivatives, especially in the context of open-weight model reuse, where unauthorized derivations and false lineage claims are more prevalent.

\section*{Acknowledgments}
We thank the anonymous reviewers and our shepherd for their
constructive comments. This work was  supported
by the Strategic Priority Research Program of the
Chinese Academy of Sciences (NO.XDB0690302),  NSFC under grant No.62371450 and Project of Key Laboratory of Cyberspace Security Defense under grant No.2025-A03.

%
%

\section*{Ethical Considerations}
\textcolor{blue}{
In this work, we study the problem of model lineage attestation using both publicly released models and those we trained on open-source datasets. All experiments are carried out entirely within local premises on controlled servers, ensuring that no external systems or third-party services were involved. Nevertheless, we emphasize that any future application of lineage attestation techniques to proprietary or commercial models should be conducted only with explicit authorization and in accordance with relevant privacy policies and regulatory standards.
}


\textcolor{blue}{
Lineage attestation introduces considerations for several stakeholders: model developers and model owners may be affected when derivation relationships are inferred, as such findings can influence intellectual-property claims or assertions of model provenance; organizations deploying ML systems may rely on lineage conclusions to support auditing or compliance decisions, and thus have an interest in ensuring that such analyses are conducted appropriately; end-users and the wider public are indirectly connected, since the trustworthiness of deployed AI systems depends in part on the authenticity of their underlying models.}

\textcolor{blue}{
As with any analysis technique that examines model parameters, it is plausible that lineage attestation could be misused to infer relationships between models without authorization. By determining whether one model is derived from another, an actor could potentially employ such information to support unauthorized provenance claims or to dispute ownership inappropriately. However, these risks are limited in practice: lineage attestation requires explicit access to both models and does not reveal training data, internal parameters, or other sensitive content. Under white-box access conditions, the risks are further mitigated because the evaluation environment and model access paths are strictly controlled. Within these constraints, the defensive value of lineage attestation outweighs its potential risks, as it enables the reliable identification of illicit derivatives and prevents false lineage assertions in ML supply chains. The process involves comparing model parameters or derived representations, which underscores the importance of performing such evaluations only with appropriate authorization, particularly when proprietary systems are involved.
}

\section*{Open Science}
Implementations of our method as well as the model zoo derived from our study, will be made publicly available. Code available at: https://zenodo.org/records/17958570. 

\bibliographystyle{plain}
\bibliography{sample-base}


%
%
%
%
%
%


\appendix
\section{Implementation Details for Building Diverse Model Families}
\label{detailsetup}
\textbf{Implementation details for building MobileNet-V2 model families.}
We also constructed a model family based on the MobileNet-V2 architecture~\cite{mobilenet}. In this family, the parent and child models were trained on different datasets to simulate cross-dataset transfer experiments.   We designed multi-generation fine-tuning paths such as Caltech-101\cite{caltech101} → Tiny ImageNet\cite{tinyimg} → MixedDataset, Flowers\cite{flowers} → CIFAR100\cite{cifar100} → MixedDataset, and Tiny ImageNet → Oxford-IIIT Pet\cite{pet} → MixedDataset. Here, MixedDataset refers to a set of diverse fine-tuning tasks randomly selected from Aircraft\cite{aircraft}, DTD\cite{dtd}, and Dogs\cite{standogs}, in order to increase domain diversity. All models in this family were developed using the same randomization strategies as in the ResNet-18 model library.

\begin{figure*}[htbp]
	\centering
\includegraphics[width=0.95\linewidth]{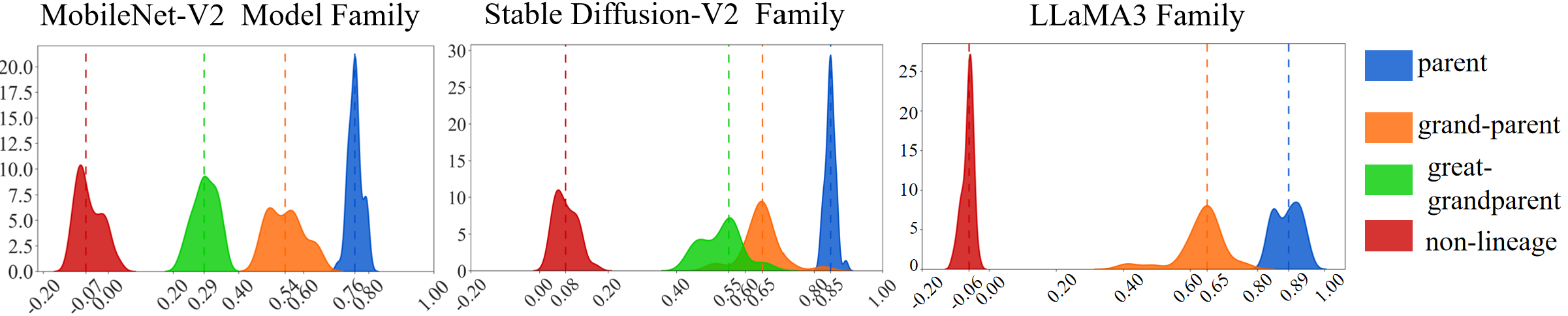}
	\caption{ Kernel density estimation(KDE)
		of similarity scores across different model architectures. Each subplot shows the similarity distribution between a child model and its parent, grandparent, great-grandparent and non-lineage items, respectively.}
	\label{diedaiarchi}
\end{figure*}

\textbf{Implementation Details for Building Diffusion Model Families.} 
For fine-tuning via DreamBooth, we adopt the official dataset provided by the authors, utilizing all 30 sets of image sets. The fine-tuning is performed using the officially recommended hyperparameters to ensure reproducibility and consistency with standard usage. More specifically, the learning rate was set to $1\times10^{-5}$. During the fine-tuning process, we save model checkpoints at both 10,000 and 20,000 training steps to serve as representative fine-tuned versions for subsequent evaluation and analysis. For LoRA-based fine-tuning, we employ the publicly available \textit{Naruto}~\cite{naruto} dataset and explore two LoRA rank settings: 128 and 256, following common practices in diffusion model adaptation tasks. The learning rate is set to $1\times10^{-5}$, and we save the model checkpoint after 10,000 training steps as the fine-tuned version for subsequent evaluation.

\textbf{Implementation Details for Building LLM  Model Families.} 
We randomly selected a subset of models under \textit{meta-llama/Llama-3.1-8B} that serve as parent models, each having at least one identifiable descendant model—most of which are derived via fine-tuning. 
 For \textit{nvidia/OpenMath2-Llama3.1-8B}, since only the merged model is publicly available on Hugging Face, we treat the merged version as its derivation in our analysis.
We randomly select a subset of models under \textit{Qwen/Qwen2.5-1.5B} that serve as parent models,  covering a diverse set of fine-tuned variants, including those derived via instruction tuning, model merging, and adapter-based fine-tuning. This setup ensures that our evaluation spans a broad spectrum of real-world adaptations commonly applied to Qwen-based models.  These descendant model variants, covering both LLaMA and Qwen families, are available in our OpenScience repository.\footnote{\url{https://zenodo.org/records/17958570.}}
Based on their Hugging Face model descriptions, we selected the datasets ``MMLU'', ``LLMTwin'', ``OpenMath'', ``Guard'', ``Tulu'', and ``Dolphin'' to construct probing datasets, from which we randomly sampled 50 prompts per dataset and collected 5 responses per prompt to capture the models' knowledge. Due to experimental constraints, we limited our evaluation to LLMs within Hugging Face model repository.
 Since both grandparent–child and sibling relationships involve models that are separated by two generations of fine-tuning datasets, we approximate the similarity of grandparent–child pairs using that of sibling models. This substitution is reasonable because sibling models, though derived from different fine-tuning paths, reflect a comparable degree of lineage distance.

\section{Additional Experiments}
\subsection{Fine-Grained Analysis of the Knowledge Similarity} 
\label{ftroundsandarchi}
\textbf{Effect of Model Architecture.} 
 Fig.~\ref{diedaiarchi} shows the KDE curves of model families based on different network architectures. 
 To demonstrate the effectiveness of our method in realistic model repository scenarios, we report validation results only on parent–child, non-lineage, and grandparent–child relationships in LLaMA3 family.
 As the lineage relationship becomes weaker, the similarity between knowledge vectors gradually decreases. Although in Diffusion experiments, there is no obvious margin between grandparent and great-grandparent relationship, we believe this is mainly due to two reasons. First, the randomly sampled probe dataset does not fully cover the semantic scope of the parent model, and therefore yields highly similar features for those two generations. Second, large foundation models trained on diverse data sources usually retain most of their original semantic representations after fine-tuning, which reduces the semantic gap between successive generations.
 
\textbf{Effect of Fine-tuning Rounds.} 
We define one generation as fine-tuning a model once with a downstream dataset. As discussed in the main text, the lineage similarity score predicted by our framework decreases as the number of fine-tuning generations increases. To systematically evaluate this effect, we conduct experiments on both MobileNet-V2 and ResNet-18 families, where models are derived from a single source and thus more sensitive to fine-tuning transformations.

We consider two scenarios: (\emph{i}) \textbf{same-domain fine-tuning}, where the original model is trained on a subset of Caltech101 classes and the downstream fine-tuning tasks are sampled from other Caltech101 categories; and (\emph{ii}) \textbf{cross-domain fine-tuning}, where the downstream tasks are drawn from the Dogs dataset. As shown in Fig.~\ref{diedai}, the similarity score consistently decreases with the number of fine-tuning rounds, confirming that iterative fine-tuning progressively obfuscates the lineage attestation trace. Importantly, the scores remain distinguishable from those of non-lineage models: for MobileNet-V2, the score stays above $0.2$, and for ResNet-18, it remains above $0.3$. This indicates that while fine-tuning reduces knowledge consistency, the lineage signal persists across multiple generations. Moreover, the difference between same-domain and cross-domain settings is marginal, suggesting that the depth of fine-tuning has a stronger impact on lineage verification than task distribution shifts.

\begin{figure}[htbp]
	\centering
	\includegraphics[width=0.9\linewidth]{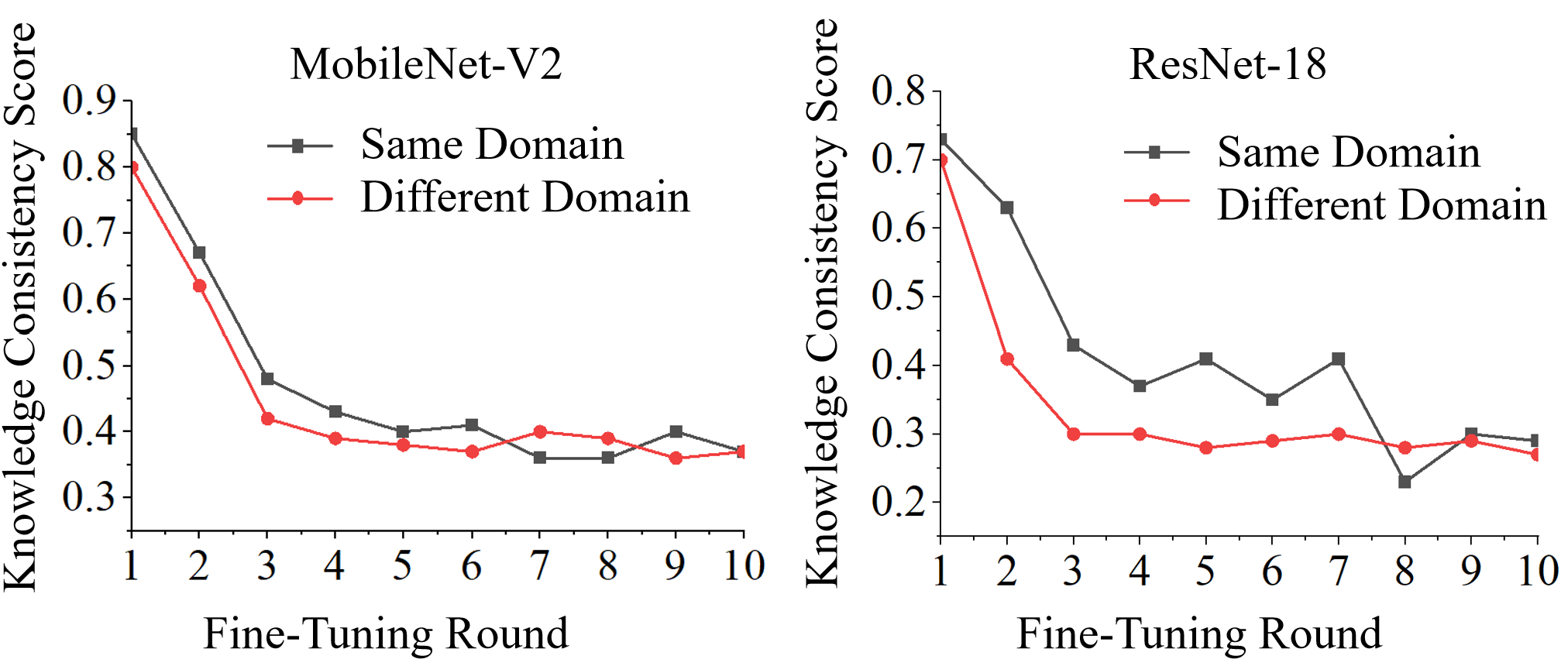}
	\caption{ Knowledge consistency scores with increasing rounds of fine-tuning.}
	\label{diedai}
\end{figure}



\subsection{\textcolor{blue}{Additional Ablation Studies and Details}} 
\label{probealterdadta}
\textcolor{blue}{
\textbf{Detailed implementation under Trusted Alternative Probe Datasets.} For in-domain probing, we resample a
new set of data points from the original training distribution to form an
alternative probe set, ensuring that the semantic domain remains
consistent while the specific samples differ. For cross-domain probing, we replace the original probe dataset with a
trusted dataset drawn from a related but distinct distribution—for example, using ImageNet as a substitute for Caltech-101 in classification
models, using CIFAR-100–derived captions as a replacement for the
COCO captions in diffusion models, or using} \verb|mgsm, gsm8k, hellaswag, humaneval|, \verb|arc_challenge| as substitutes for the original probe prompts in LLMs.

\textbf{Effect of the Key Component in Our Framework.} 
We ablate each component by measuring the TPR drop upon removal, evaluating its
impact on lineage discrimination. Results of  MobileNet-V2  family are
shown in Table~\ref{ablation}.
Firstly, removing the evolution knowledge vector(w/o $\mathbf{h}_\Delta$) reduces the score to 0.66, confirming that the evolution model effectively captures parent-to-child knowledge transfer. 
Secondly, substituting the knowledge encoder with mean pooling  (w/o $\Psi$) of feature embeddings yields a TPR of 0.91 but incurs a 0.08 accuracy drop, indicating that $\Psi$ compactly encodes lineage-relevant knowledge beyond raw features.
Finally, replacing the knowledge fusion network (w/o $\Phi$) with a simple sum $\mathbf{h}_C + \mathbf{h}_\Delta$ decreases performance from 0.99 to 0.43, demonstrating $\Phi$’s role in integrating inherited and evolved knowledge.
\begin{table}[htbp]
	\caption{Ablation study for the proposed framework, where w/o refers to without. }
	 \vspace{-0.1cm}
    \centering
    \begin{adjustbox}{max width =0.45\textwidth}
	\begin{tabular}{c|c|c|c|c}
		\toprule
		 &our MLA &w/o~$\mathbf{h}_\Delta$ &w/o~$\Psi$ & w/o~$\Phi$ \\
		\midrule
		TPR&0.99 &0.66 & 0.91 &0.43 \\
		\bottomrule
	\end{tabular}
	    \end{adjustbox}
    \label{ablation}
\end{table}

\subsection{\textcolor{blue}{Extended Adaptive Attack}} 
\textcolor{blue}{\textbf{Knowledge Infusion Attack.} In the False Lineage
Claim Rejection scenario,  the adversary may attempt to construct a forged parent
model $f_P^\mathcal{A}$ that exhibits behavior similar to the child model \( f_C \).  To achieve this, the adversary queries
\( f_C \) on the probe dataset and uses its responses to fine-tune
\( f_P^{\mathcal{A}} \), thereby distilling probe-induced behavior into the
forged model to fabricate a plausible parent–child
relationship. As shown in Table~\ref{knowledgeinfusion}, its similarity to \( f_C \)
remains low (0.31/0.31/0.32 for 6\%, 30\%, and 60\% probe usage),
well below the 0.7 threshold. These results indicate that iterative fine-tuning progressively distorts the alignment between parameter evolution and underlying knowledge, leading to
decreasing lineage similarity instead of successful lineage reconstruction.
}
\begin{table}[htbp]
\centering
\caption{Resistance to Knowledge Infusion Attacks.}
\label{tab:false_lineage}
\begin{tabular}{c c c}
\toprule
\textbf{Probe Usage} & \textbf{Similarity to $f_C$} & \textbf{Verified($T$=0.7)} \\
\midrule
6\%   & 0.31 & $\times$ \\
30\%  & 0.31 & $\times$\\
60\%  & 0.32 & $\times$ \\
\bottomrule
\end{tabular}
\label{knowledgeinfusion}
\end{table}

\section{\textcolor{blue}{Discussion  and Limitations}}
\textcolor{blue}{
\textbf{Practical considerations for  third-party evidence collection.}
MLA operates under a   model management platform in which a trusted attestation authority evaluates artifacts submitted by the claimant. This requirement aligns with existing model-governance practices: enterprise registries and public model platforms already maintain model parameters, version histories, and provenance
metadata, making the forwarding of platform-certified artifacts—such as the initialization $\theta_0$—a natural extension of current workflows.
Under this setting, authenticity of the submitted model parameters can be ensured, as the malicious claimant cannot modify platform-certified checkpoints.
However, verifying the authenticity of auxiliary evidence such as probe datasets is more challenging. MLA focuses on validating claimed lineage relationships rather than reconstructing the primordial training origin of a model, and a trusted authority may not be able to guarantee that a submitted probe dataset is the exact one used during fine-tuning.  
In such cases, the authority may instead construct alternative probe sets drawn from semantically similar
sources. As demonstrated in Fig.~\ref{probe}, semantically aligned probe datasets achieve comparable attestation performance; however, developing stronger mechanisms for verifying dataset provenance remains an open direction for future work.}


\end{document}